\documentclass[nonacm, sigconf, screen]{acmart}

\usepackage{multirow}
\usepackage{array}
\usepackage{xcolor}
\usepackage{subcaption}
\usepackage{enumitem}
\usepackage{tabularx}

\usepackage{balance}

\usepackage{tikz}
\usetikzlibrary{arrows.meta, positioning, decorations.markings}

\usepackage[ruled, vlined, linesnumbered]{algorithm2e}
\SetKw{Return}{return}

\usepackage{mathrsfs}
\usepackage{amsmath}
\usepackage{amsthm}

\usepackage{dashrule}

\usepackage[normalem]{ulem}

\newtheorem{definition}{Definition}

\newtheorem{theorem}{Theorem}
\newtheorem{corollary}{Corollary}
\newtheorem{proposition}{Proposition}
\newtheorem{example}{Example}

\newcommand{\latent}{{\sf LatentPre}}

\begin{document}

\title{Fair Data Pre‑Processing with Imperfect Attribute Space}

\author{Ying Zheng}
\affiliation{%
  \institution{National University of Singapore}
  \country{Singapore}}
\email{zheng.ying@u.nus.edu}

\author{Yangfan Jiang}
\affiliation{%
  \institution{National University of Singapore}
  \country{Singapore}}
\email{jyangfan@u.nus.edu}

\author{Kian-Lee Tan}
\affiliation{%
  \institution{National University of Singapore}
  \country{Singapore}}
\email{tankl@comp.nus.edu.sg}

\begin{abstract}

Fair data pre-processing is a widely used strategy for mitigating bias in machine learning. A promising line of research focuses on calibrating datasets to satisfy a designed fairness policy so that sensitive attributes influence outcomes only through clearly specified legitimate causal pathways. While effective on clean and information-rich data, these methods often break down in real-world scenarios with imperfect attribute spaces, where decision-relevant factors may be deemed unusable or even missing. To address this gap, we propose {\sf LatentPre}, a novel framework that enables principled and robust fair data processing in practical settings. Instead of relying solely on observed attributes, {\sf LatentPre} augments the fairness policy with latent attributes that capture essential but subtle signals, enabling the framework to operate as if the attribute space were perfect. These latent attributes are strategically introduced to guarantee identifiability and are estimated using a tailored expectation-maximization paradigm. The raw data is then carefully refined to conform to this latent-augmented policy, effectively removing biased patterns while preserving justifiable ones. Extensive experiments demonstrate that {\sf LatentPre} consistently achieves strong fairness-utility trade-offs across diverse scenarios, advancing practical fairness-aware data management.

\end{abstract}

\maketitle

\section{Introduction} \label{sec: intro}

Machine learning (ML) has become a core component of modern data-driven systems, supporting decision-making in domains such as finance~\cite{finance2007}, education~\cite{admission1975, pals2024algorithms}, healthcare~\cite{obermeyer2019dissecting}, and public administration~\cite{policing2017, crime2009, recommend2022}. However, these advances also raise growing concerns about fairness. In practice, ML models often unintentionally learn biased patterns from the raw data maintained by database management systems (DBMSs), and these patterns may be further amplified during training, resulting in unfair or discriminatory predictions.~\cite{lin2024mitigating, zhang2025efficient, yang2024noninvasive, zheng2024fairgen, pradhan2022interpretable, surve2025explaining, tsioutsiouliklis2021fairness, tae2024falcon}.

A promising way to mitigate this issue is to apply fairness-aware data pre-processing before the data enters ML training. This ensures that models are trained on data as if collected in a fair world, rather than drawn from potentially biased distributions stored in DBMSs. Substantial research efforts have been devoted to exploring fairness criteria and developing fair data pre‑processing solutions~\cite{salimi2019interventional, pirhadi2024otclean, zheng2025causalpre, salazar2021automated, galhotra2022causal, feldman2015certifying, nabi2018fair, zemel2013learning, calmon2017optimized, gordaliza2019obtaining, zhang2023iflipper, lahoti2019ifair}. Among them, methods based on \emph{justifiable fairness}~\cite{salimi2019interventional} have received considerable attention~\cite{pirhadi2024otclean, zheng2025causalpre, salazar2021automated, galhotra2022causal, pujol2023prefair, zuo2024interventional}, due to their favorable fairness-utility trade-offs, interpretability, and practical applicability. 

Justifiable fairness restricts sensitive attributes, such as gender or race, to influencing decisions only through ethically acceptable attributes, known as admissible attributes, even if such attributes carry biased information. In contrast, other attributes that encode bias, termed inadmissible attributes, are not allowed to affect decisions. To instantiate this notion in data pre-processing, a standard recipe is to define a \emph{fairness policy} that models a hypothetically fair world, in which only justifiably fair causal relationships remain. The raw data is then modified so that its empirical distribution conforms to this policy, as if collected from such a world.  Existing methods~\cite{salimi2019interventional, zheng2025causalpre} typically construct this policy within an attribute space that mirrors the raw DBMS schema, relying on straightforward fair representations. When applied to clean and information‑rich datasets with clearly specified admissible and inadmissible attributes, these frameworks achieve desirable fairness-utility trade-offs.  

\begin{figure}[t]
    \centering
    \includegraphics[width=0.98\linewidth, keepaspectratio]{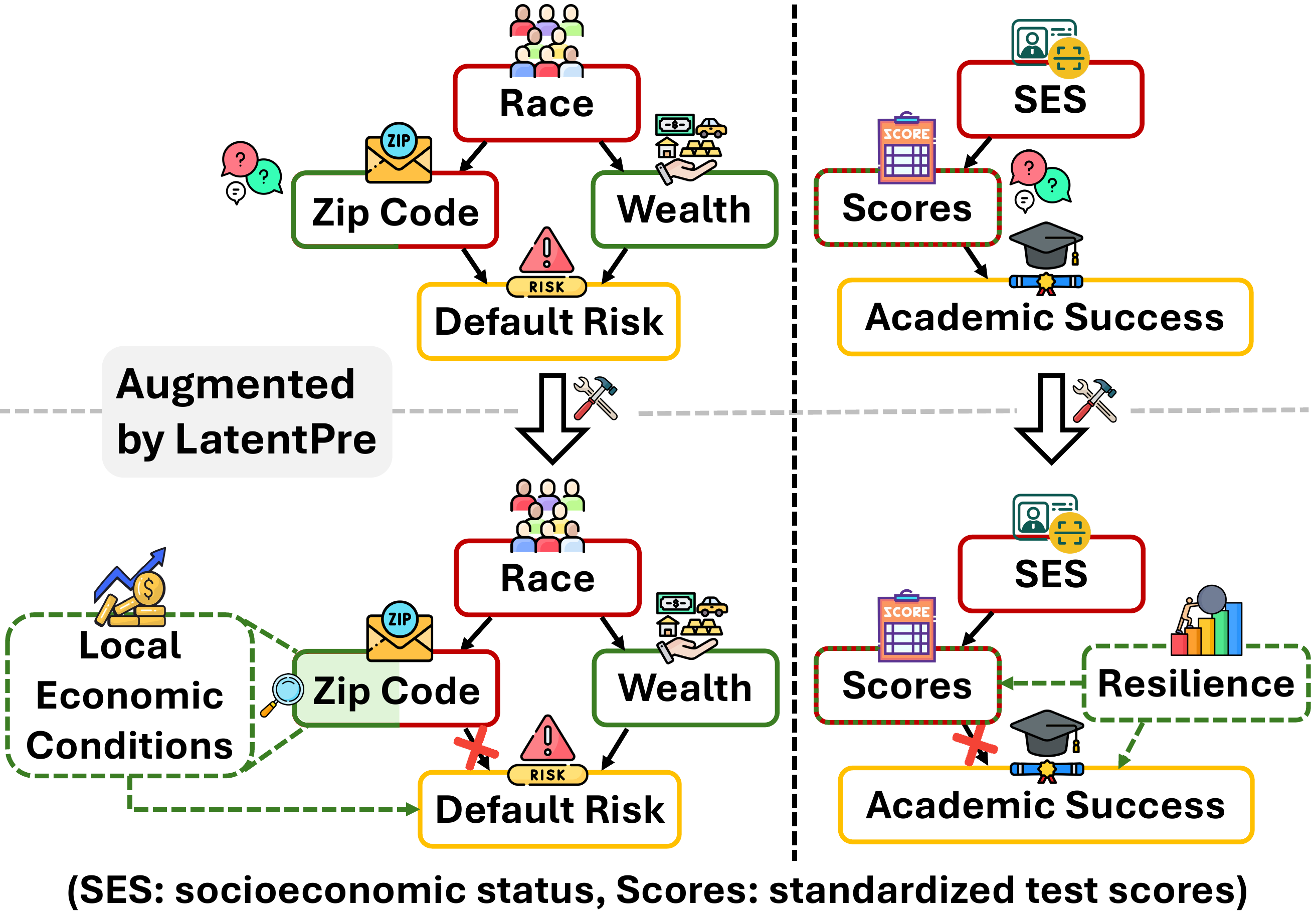}
    \caption{Attribute graphs of real-world examples.}
    \label{fig: example}
\end{figure}

\subsection{Fair Data Pre-Processing in the Wild}\label{subsec:intro-motivation}

In real-world DBMSs, however, raw data often comes with an imperfect attribute space. One common challenge in such settings is \textbf{attribute ambiguity}, where some attributes have unclear roles, making it difficult even for domain experts to decide whether they should be treated as admissible or inadmissible. When such attributes are conservatively treated as inadmissible to avoid potential bias, existing methods typically block their influence on the decision. This operation can substantially undermine data utility, as valuable predictive signals are removed during processing. Another challenge is \textbf{attribute absence}, where certain utility-relevant factors are not recorded in the schema due to collection constraints or their abstract nature. Although the effect of such factors may be partially reflected in observed attributes, the factors themselves are not explicitly present. When the corresponding information is entangled with inadmissible attributes, existing methods may also suppress the resulting signals, even when these signals are legitimate. As a result, the already limited data utility can be further compromised after fairness-aware processing, and the processed data may no longer support effective decision-making. To make these challenges concrete, we illustrate both sources of imperfection with the following examples.
\begin{example}[Attribute Ambiguity] \label{eg-1}
    Consider a loan default prediction task. As shown in Figure~\ref{fig: example} (top left), race is sensitive, wealth is admissible, and default risk is the outcome. The role of zip code, however, is ambiguous. It correlates with race and may therefore reflect historical or structural bias, while it can also serve as a proxy for local economic conditions relevant to an applicant's financial stability. For example, in a town whose economy depends on a single large factory, news of the factory's impending shutdown can sharply increase residents’ future financial uncertainty. In this case, zip code becomes a forward-looking risk indicator beyond what current wealth alone can reveal. If practitioners conservatively treat zip code as inadmissible, existing fairness-aware methods typically block their influence on the decision, which can reduce data utility.
\end{example}
\begin{example}[Attribute Absence] \label{eg-2}
    Consider a university scholarship allocation task that aims to predict academic success. As shown in Figure~\ref{fig: example} (top right), socioeconomic status (SES) is sensitive, academic success is the outcome, and standardized test scores (scores) are often deemed inadmissible because they are heavily influenced by SES. Students from high-SES backgrounds typically benefit from better preparation resources and repeated test attempts, so high scores do not always reflect ability alone.
    For students from low-SES backgrounds, a high score is often achieved under substantial challenges and may therefore reflect resilience, an unobserved trait that is highly predictive of academic success.
    Because resilience is not recorded in the data, existing methods cannot capture it, thereby overlooking such students’ potential and weakening prediction.
\end{example}
As these examples show, imperfections are inherent in real-world DBMS data. Existing fair data pre-processing methods are not designed for such cases and can fail when applied in practice, as shown in Section~\ref{subsec: exp-endtoend}.
Simple remedies such as imputing missing values or applying data augmentation do not account for fairness and may even amplify existing biases~\cite{pessach2022review}.
These limitations reveal a critical barrier to practically deploying fair data pre-processing in DBMSs and motivate the need for robust approaches that can operate under imperfect attribute spaces while still ensuring justifiable fairness and preserving data utility.

\subsection{Contributions} \label{subsec: intro-contribution}

Motivated by this, we present \latent{}, a justifiably fair data pre-processing framework that can robustly capture informative signals from raw data to support reliable and effective downstream decision-making, even in DBMS environments with an imperfect attribute space. To achieve this, \latent\ follows the standard \emph{policy-then-adjust} paradigm, while departing from prior solutions through a fundamentally different policy formulation and a corresponding adjustment procedure. The key idea is to augment the hypothetically fair world with carefully introduced \emph{latent attributes} that represent essential but unrecorded factors, rather than constructing it solely from the raw DBMS attribute space. This latent augmentation compensates for practical imperfections in the observed schema and enables \latent\ to recover and leverage predictive signals in a manner consistent with justifiable fairness.

We revisit Examples~\ref{eg-1} and~\ref{eg-2} to illustrate this key idea. In Example~\ref{eg-1}, \latent\ introduces a latent attribute, as shown in the bottom-left panel of Figure~\ref{fig: example}, to represent local economic conditions associated with zip code, thereby retaining forward-looking information about financial uncertainty. In Example~\ref{eg-2}, \latent\ similarly introduces a latent attribute, as shown in the bottom-right panel of Figure~\ref{fig: example}, to represent resilience, a legitimate factor underlying test scores that is otherwise absent from the schema.
In both cases, \latent\ blocks the unfair pathway from the inadmissible attribute to the label and preserves utility through a latent attribute that is independent of sensitive attributes. As a result, it can reliably use potential decision-relevant information even when the attribute space is imperfect.

\vspace{2mm}
\noindent\textbf{Challenges.}
While latent attributes are a well-established modeling tool, incorporating them into fair data pre-processing, especially under an imperfect attribute space, presents two significant challenges. First, determining how to introduce latent attributes and model their relationships with observed attributes is non-trivial, as they must reliably capture fair yet decision-relevant signals that are often only implicitly encoded in the raw data via inadmissible attributes. Second, once the policy includes unobserved latent attributes, existing fair data adjustment techniques no longer apply.
To our knowledge, no prior work has formally addressed these challenges in the pre-processing setting.

\vspace{2mm}
\noindent\textbf{Technical overview.}
\latent{} is built on a set of tightly integrated technical components designed to address the above challenges. At the policy level, it augments the fair world with latent attributes $L$ that influence a strategically selected subset of {inadmissible attributes $\mathcal{I}_c$} and the label $Y$, while remaining independent of sensitive attributes $\mathcal{S}$. Intuitively, $L$ models valuable predictive factors that are not recorded in the explicit schema but are indirectly reflected in the raw data through inadmissible attributes, which prior solutions cannot fully exploit without violating fairness constraints. This policy design enables \latent\ to uncover and leverage such legitimate signals while maintaining justifiable fairness. 

However, a naive approach may yield little benefit if $L$ is poorly placed or modeled. Without careful design, the policy formulation may not be provably effective, and the latent attributes could remain uninformative, introducing arbitrary or redundant structure. In essence, modeling $L$ requires enough observable constraints to uniquely determine its effect, analogous to solving for unknowns requiring enough independent equations. In latent causal models, this requirement is captured by \emph{identifiability}~\cite{allman2009identifiability}. To avoid such degeneracy, we enforce identifiability through a principled placement strategy, which provably guarantees that the inserted latent attributes capture genuine decision‑relevant signals.  Specifically, our strategy positions $L$ to intervene on $\mathcal{I}_c$ and $Y$, combined with a local pruning mechanism that minimally reshapes the causal structure surrounding $\mathcal{I}_c$ to meet identifiability conditions.

A remaining practical concern is that the augmented policy contains unobserved attributes, making direct data adjustment infeasible. To address this, we formulate the data adjustment process as a parameter estimation problem involving latent attributes, and efficiently solve it using an expectation‑maximization (EM) paradigm tailored to the structure of our policy. Once the parameters are estimated, the optimal empirical distribution conforming to the policy is determined, and the raw data can then be adjusted accordingly to complete the fairness-aware processing.

Extensive experiments verify that \latent\ consistently achieves strong performance across diverse scenarios, demonstrating robustness to imperfect and complex real‑world data environments.

\vspace{2mm}
\noindent\textbf{Roadmap.}
The remainder of this paper is organized as follows. Section~\ref{sec: preliminary} introduces the required notions and background. Section~\ref{sec: problem} formalizes the fair data pre-processing problem under an imperfect attribute space. Sections~\ref{sec: methodology} and~\ref{sec: component} then present our proposed solution, \latent. Section~\ref{sec: exp} provides a comprehensive evaluation, examining the behavior of \latent\ across various scenarios. Finally, Section~\ref{sec: related} reviews related work, and Section~\ref{sec: conclusion} concludes the paper.

\section{Preliminaries} \label{sec: preliminary}

\subsection{Causal DAGs} \label{subsec: pre-causal-dags}
A causal DAG $\mathcal{G}$ is a directed acyclic graph over attributes $\mathcal{V}=\{V_1,\ldots,V_d\}$ that encodes direct causal relationships among them. Each node corresponds to an attribute (viewed as a random variable in the underlying distribution), and each directed edge represents a direct causal effect. For each $V_i$, let $\Pi_i$ denote its parent set in $\mathcal{G}$. The model associates $V_i$ with a conditional distribution $\mathbb{P}[V_i \mid \Pi_i; \theta_i]$, where $\theta_i$ parameterizes the conditional probability table (CPT) of $V_i$ given $\Pi_i$. The joint distribution over $\mathcal{V}$ then admits the standard Bayesian network factorization~\cite{pearl2009causality}:
\begin{align} \label{eq: bayes}
    \mathbb{P} \left[ \mathcal{V}; \theta \right] = \mathbb{P} \left[ V_1, \dots, V_d; \theta \right] 
    = \prod_{i=1}^{d} \mathbb{P} \left[ V_i \mid \Pi_i; \theta_i \right],
\end{align}
where $\theta = (\theta_1, \dots, \theta_d)$ collects the parameters for all CPTs.

Beyond factorization, causal DAGs encode conditional independencies through d-separation~\cite{pearl2009causality}. For disjoint node sets $X,Y,Z\subseteq \mathcal{V}$, $X$ and $Y$ are \emph{d-separated} by $Z$ in $\mathcal{G}$ if every path between $X$ and $Y$ is blocked when conditioning on $Z$. A distribution $\mathbb{P}$ over $\mathcal{V}$ is \emph{Markov compatible} with $\mathcal{G}$ if every d-separation in $\mathcal{G}$ implies the corresponding conditional independence in $\mathbb{P}$. Conversely, $\mathbb{P}$ is \emph{faithful} to $\mathcal{G}$ if every conditional independence in $\mathbb{P}$ is implied by a d-separation in $\mathcal{G}$. Following prior work~\cite{salimi2019interventional, galhotra2022causal, pujol2023prefair}, we assume that the observed distribution $\mathbb{P}$ is Markov compatible with and faithful to the underlying causal DAG.

Causal DAGs support the definition of justifiable fairness~\cite{salimi2019interventional} (introduced in the next section) through interventions, represented by the \textit{do} operator~\cite{pearl2009causality}. For an attribute $X$ and value $x$, the intervention $\text{\textit{do}}(X=x)$ forces $X$ to take value $x$ and severs all incoming causal influences into $X$. The resulting interventional distribution $\mathbb{P}[O \mid \text{\textit{do}}(X=x)]$ characterizes the causal effect of $X$ on an outcome $O$. Graphically, $\text{\textit{do}}(X=x)$ corresponds to removing all incoming edges to $X$ and fixing its value.

\subsection{Justifiable Fairness}

Consider a classifier $\mathcal{M}: \text{Dom}(\mathcal{X}) \to \text{Dom}(O)$ that maps input features $\mathcal{X}$ to an outcome variable $O$, where $\text{Dom}(\cdot)$ denotes the domain. Let $\mathcal{S}\subseteq \mathcal{X}$ represent the set of sensitive attributes; then, justifiable fairness is defined as follows. 
\begin{definition}[$\mathcal{K}$-fairness~\cite{salimi2019interventional}]
    For a subset $\mathcal{K} \subseteq \mathcal{X} \setminus \mathcal{S}$, we say that $\mathcal{M}$ is \emph{$\mathcal{K}$-fair} with respect to the sensitive set $\mathcal{S}$ if, for every assignment $\mathcal{K}=\kappa$, intervening on $\mathcal{S}$ does not affect the distribution of the outcome once $\mathcal{K}$ is fixed:
    \begin{align*} 
        &\mathbb{P} \left[ O = o \mid \text{do}(\mathcal{S} = \text{\scriptsize$\mathcal{S}$}_0), \text{do}(\mathcal{K} = \text{\scriptsize$\mathcal{K}$}) \right] \notag \\
        = \;&\mathbb{P} \left[ O = o \mid \text{do}(\mathcal{S} = \text{\scriptsize$\mathcal{S}$}_1), \text{do}(\mathcal{K} = \text{\scriptsize$\mathcal{K}$}) \right].
    \end{align*}
\end{definition}
\begin{definition}[Justifiable Fairness~\cite{salimi2019interventional}] \label{def: justifiable}
    Let $\mathcal{A}\subseteq \mathcal{X}$ denote the collection of \emph{admissible attributes}, i.e., features that are permitted to influence $O$ even if they are themselves affected by sensitive attributes. A classifier $\mathcal{M}$ is said to satisfy \emph{justifiable fairness} if it is $\mathcal{K}$-fair for every $\mathcal{K}$ with $\mathcal{A} \subseteq \mathcal{K} \subseteq \mathcal{X}$.
\end{definition}

To isolate the effect of data pre-processing from that of model training, we follow prior work~\cite{salimi2019interventional, galhotra2022causal,zheng2025causalpre} and assume a \emph{reasonable} classifier, i.e., one that closely approximates the data distribution on which it is trained.
The following corollary and proposition provide sufficient structural conditions on the data under which any reasonable classifier satisfies justifiable fairness.
\begin{corollary}[\cite{salimi2019interventional}] \label{cor:justifiable}
    Let $\mathcal{G}$ be the attribute graph derived from a dataset. If all causal pathways from sensitive attributes to the label pass through at least one admissible attribute, then any reasonable classifier trained on this dataset is regarded as justifiably fair.
\end{corollary}
\begin{proposition}[\cite{salimi2019interventional,zheng2025causalpre}] \label{prop: fair}
    Let $\mathcal{D}$ be a database instance and $\mathcal{G}$ its attribute graph over the attribute set $\mathcal{V}$. Under attribute specification, $\mathcal{V}$ is fully partitioned into five disjoint subsets: $\mathcal{S}$ (sensitive), $\mathcal{I}$ (inadmissible), $\mathcal{A}$ (admissible), $\mathcal{W}$ (additional), and $\mathcal{Y}$ (label). Here $\mathcal{W}$ refers to additional attributes that are neither sensitive nor causally relevant to fairness. If, in $\mathcal{G}$, every edge directed toward the label attribute in $\mathcal{Y}$ originates exclusively from attributes in $\mathcal{A} \cup \mathcal{W}$, then any reasonable classifier trained on $\mathcal{D}$ respects justifiable fairness. Equivalently, the condition holds whenever the parent set $\Pi$ of the label attribute is fully contained in $\mathcal{A} \cup \mathcal{W}$, i.e., $\Pi \subseteq \mathcal{A} \cup \mathcal{W}$.
\end{proposition}

\subsection{Latent Model and Identifiability}

Causal DAGs offer a principled way to model dependencies among observed attributes, but in many cases, there exist hidden or unmeasured factors that also affect the outcomes. To consider such factors, let $L$ denote an unobserved latent attribute in addition to the observed attributes $\mathcal{V}$. A latent causal model then specifies a joint distribution over $\left( \mathcal{V}, L \right)$, written as $\mathbb{P} \left[ \mathcal{V}, L; \theta \right]$, where $\theta = \left( \theta_1, \dots, \theta_d, \theta_l \right)$ denotes the model parameters. The observed distribution is obtained by marginalizing out $L$: 
\begin{align*}
    \mathbb{P}\left[ \mathcal{V}; \theta \right] = \sum_{L}{\mathbb{P}\left[ \mathcal{V}, L; \theta \right]},
\end{align*}
where the summation is taken over all possible states of $L$.

Since $L$ is unobserved, its role must be inferred from the observed data. This naturally raises the question of whether the model parameters can be reasonably determined from the observed data without ambiguity, a property known as \emph{identifiability}. When identifiability holds, the estimated parameter $\theta$ faithfully represents the true data-generating process, enabling reliable causal reasoning. The concept is formally defined as follows.
\begin{definition}[Identifiability~\cite{allman2009identifiability}]
    Consider a latent causal model parameterized by $\theta \in \Theta$, where $\Theta$ is the parameter space.
    Let $\mathbb{P}[\mathcal{V},\mathcal{L};\theta]$ denote the joint distribution induced by $\theta$, and let $\mathbb{P}[\mathcal{V};\theta]$ denote the corresponding observed distribution obtained by marginalizing out $\mathcal{L}$.
    The model is said to be \emph{identifiable} if \[ \forall \theta, \theta' \in \Theta, \quad \mathbb{P}[\mathcal{V}; \theta] = \mathbb{P}[\mathcal{V}; \theta'] \ \Rightarrow \ \theta = \theta'. \]
    Equivalently, the mapping $\theta \mapsto \mathbb{P}[\mathcal{V};\theta]$ is injective, so the parameters are uniquely determined by the observed distribution.
\end{definition}

While strict identifiability provides a clean theoretical guarantee, it is often overly restrictive for latent causal models. In general, it fails over the full parameter space~\cite{allman2009identifiability, allman2015parameter}, because degenerate parameter configurations may render distinct models observationally indistinguishable, and inherent symmetries such as latent-state label permutations do not affect the observed distribution. This motivates weaker but practically sufficient notions of identifiability, introduced below.
\begin{definition}[Generic Identifiability~\cite{allman2009identifiability}]\label{def: generic-indentifiability}
    A latent causal model is said to be \emph{generically identifiable} if there exists a subset $\Omega \subseteq \Theta$ of Lebesgue measure zero such that \[ \forall \theta, \theta' \in \Theta \setminus \Omega, \quad \mathbb{P}[\mathcal{V}; \theta] = \mathbb{P}[\mathcal{V}; \theta'] \ \Rightarrow \ \theta = \theta'. \]
    Equivalently, the mapping $\theta \mapsto \mathbb{P}[\mathcal{V};\theta]$ is injective almost everywhere, except on a measure-zero subset of $\Theta$.
\end{definition}
\begin{definition}[Generic Identifiability up to Label Swapping~\cite{allman2009identifiability}]
    Let $\Pi$ be the set of all permutations of the state labels of $\mathcal{L}$, and for $\pi \in \Pi$, let $\pi(\theta)$ denote the parameter vector obtained by applying $\pi$ to the state labels of $\mathcal{L}$ in all relevant CPTs. 
    A latent causal model is said to be \emph{generically identifiable up to label swapping} if there exists a subset $\Omega \subseteq \Theta$ of Lebesgue measure zero such that \[ \forall \theta, \theta' \in \Theta \setminus \Omega, \quad \mathbb{P}[\mathcal{V}; \theta] = \mathbb{P}[\mathcal{V}; \theta'] \ \Rightarrow \ \exists \pi \in \Pi \text{ such that } \theta' = \pi(\theta). \]
    That is, for almost all parameter values, the observed distribution determines the parameters uniquely, up to a permutation of the latent variable’s state labels.
\end{definition}
In this paper, we focus on the notion of \emph{Generic Identifiability up to Label Swapping}, which generally does not diminish the practical utility of a model and is sufficient for most applications~\cite{allman2009identifiability, allman2015parameter}. 
For convenience, unless otherwise specified, we use the terms \emph{identifiability} or \emph{generic identifiability} interchangeably to refer to \emph{Generic Identifiability up to Label Swapping}.

\subsection{Information Measures}

We next introduce several information measures~\cite{cover1999elements} that underpin our analysis. 
Let $X$, $Y$, and $Z$ denote random variables and $\mathcal{X}$, $\mathcal{Y}$, and $\mathcal{Z}$ denote random variable sets. The entropy of $X$ measures the uncertainty associated with its distribution:
\begin{align*}
    H(X) = - \sum_{x} \mathbb{P}[X=x] \log \mathbb{P}[X=x].
\end{align*}

The conditional entropy of $X$ given $Y$ quantifies the average uncertainty remaining in $X$ once $Y$ is known:
\begin{align*}
    H(X \mid Y) = - \sum_{x,y} \mathbb{P}[X=x, Y=y] \log \mathbb{P}[X=x \mid Y=y].
\end{align*}

The mutual information (MI) between $X$ and $Y$ captures the reduction in uncertainty of one variable upon observing the other:
\begin{align*}
    I(X;Y) = H(X) - H(X \mid Y) = H(Y) - H(Y \mid X).
\end{align*}
A higher value indicates stronger dependence between $X$ and $Y$.

The conditional mutual information (CMI) between $X$ and $Y$ given $Z$ extends MI by measuring the remaining dependence between $X$ and $Y$ after accounting for $Z$:
\begin{align*}
    I(X;Y \mid Z) = H(X \mid Z) - H(X \mid Y, Z) = H(Y \mid Z) - H(Y \mid X, Z).
\end{align*}
CMI is nonnegative and equals zero if and only if $X$ and $Y$ are conditionally independent given $Z$.

For random variable sets $\mathcal{X}=\{X_1, \dots, X_m\}$ and $\mathcal{Y}=\{Y_1, \dots, Y_n\}$, the conditional mutual information between $\mathcal{X}$ and $\mathcal{Y}$ given $\mathcal{Z}$ can be decomposed using the chain rule: 
\begin{align}
    I(\mathcal{X}; \mathcal{Y} \mid \mathcal{Z}) = \sum_{i=1}^{m} \sum_{j=1}^{n} I(X_i; Y_j \mid \mathcal{X}_{<i}, \mathcal{Y}_{<j}, \mathcal{Z}), \label{eq: cmi}
\end{align}
where $\mathcal{X}_{<i}=\{X_1, \dots, X_{i-1}\}$ and $ \mathcal{Y}_{<j}=\{Y_1, \dots, Y_{j-1}\}$.

\section{Problem Statement} \label{sec: problem}

Given a database $\mathcal{D}$ with $d$ attributes $\mathcal{V}=\{V_1, \dots, V_{d-1}, Y\}$, where $Y$ denotes the label attribute, the remaining attributes $\mathcal{V}\setminus\{Y\}$ are partitioned into four disjoint sets. The set $\mathcal{S}$ contains sensitive attributes such as gender or race, and $\mathcal{I}$ includes inadmissible attributes that directly encode sensitive information. Under justifiable fairness, neither $\mathcal{S}$ nor $\mathcal{I}$ is allowed to have a direct causal effect on the decision. The set $\mathcal{A}$ contains admissible attributes whose influence on the decision is considered legitimate even if affected by $\mathcal{S}$. The set $\mathcal{W}$ includes additional attributes unrelated to sensitive information and can therefore be used as decision factors.

We consider a practical setting where the attribute space of the database $\mathcal{D}$ is imperfect due to attribute ambiguity and attribute absence, as discussed in Section~\ref{subsec:intro-motivation}. Our objective is to design a data pre-processing framework that calibrates the empirical distribution of $\mathcal{D}$ such that any reasonable downstream predictive model trained on it satisfies justifiable fairness while preserving effectiveness.

\section{LatentPre} \label{sec: methodology}

In this section, we introduce our latent-augmented framework, \latent. In a nutshell, \latent\ follows the \textit{policy-then-adjust} paradigm and operates in two stages. It first specifies a fairness policy that encodes domain knowledge and fairness requirements, and then adjusts the dataset to satisfy this policy with minimal distributional distortion while preserving data utility.

\begin{figure*}[t]
\centering
\resizebox{0.94\textwidth}{!}{%
\begin{minipage}{\textwidth}
\centering
\setlength{\tabcolsep}{0pt}

\begin{tabularx}{\textwidth}{@{}*{5}{>{\centering\arraybackslash}X}@{\hspace{10mm}}>{\centering\arraybackslash}X@{}}

\multicolumn{5}{@{}l@{}}{%
    \includegraphics[width=0.833\linewidth]{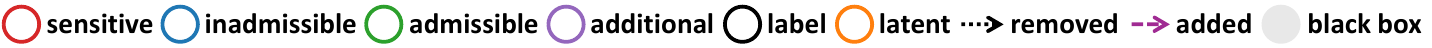}
} & \\[0.8ex]

\begin{minipage}[t]{\linewidth}
    \centering
    \includegraphics[width=\linewidth]{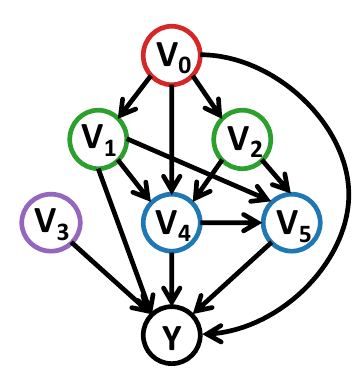}\\[-0.3ex]
    \vspace{1mm} {\small (a)}
\end{minipage}
&
\begin{minipage}[t]{\linewidth}
    \centering
    \includegraphics[width=\linewidth]{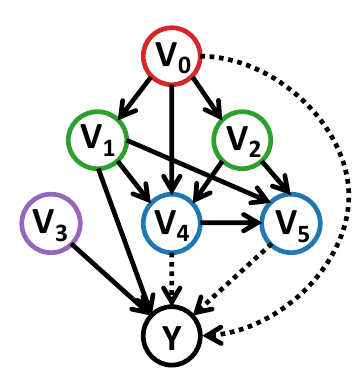}\\[-0.3ex]
    \vspace{1mm} {\small (b)}
\end{minipage}
&
\begin{minipage}[t]{\linewidth}
    \centering
    \includegraphics[width=\linewidth]{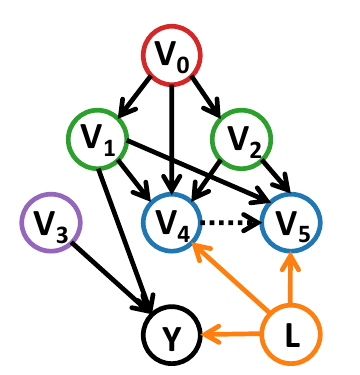}\\[-0.3ex]
    \vspace{1mm} {\small (c)}
\end{minipage}
&
\begin{minipage}[t]{\linewidth}
    \centering
    \includegraphics[width=\linewidth]{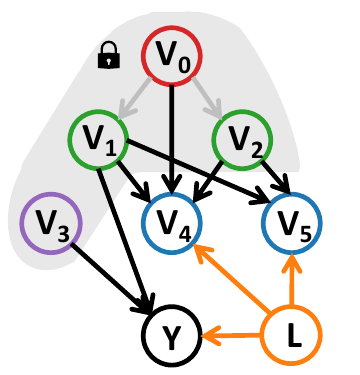}\\[-0.3ex]
    \vspace{1mm} {\small (d)}
\end{minipage}
&
\begin{minipage}[t]{\linewidth}
    \centering
    \includegraphics[width=\linewidth]{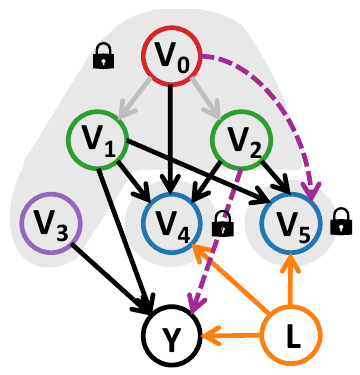}\\[-0.3ex]
    \vspace{1mm} {\small (e)}
\end{minipage}
&
\begin{minipage}[t]{\linewidth}
    \centering
    \includegraphics[width=\linewidth]{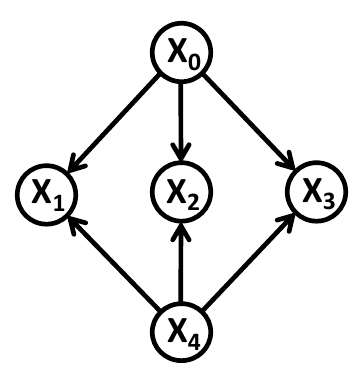}
\end{minipage}
\\

\multicolumn{5}{@{}c@{}}{%
    \begin{minipage}[t]{0.833\textwidth}
        \centering
        \vspace{-2mm}
        \captionof{figure}{Framework construction and refinement process.}
        \label{fig: reduce}
    \end{minipage}
}
&
\begin{minipage}[t]{\linewidth}
    \centering
    \vspace{-6mm}
    \captionof{figure}{The generalized model.}
    \label{fig: dag_dual}
\end{minipage}

\end{tabularx}
\end{minipage}%
}
\vspace{3mm}
\end{figure*}

\subsection{Latent-Augmented Fairness Policy} \label{subsec: method-framework}

\latent\ instantiates the policy using a causal DAG. Starting from the raw-data DAG, \latent\ refines the structure through a sequence of steps that adjust causal connections and introduce a controlled interaction between a latent attribute and observed attributes. For clarity, we temporarily assume that the raw-data DAG is available. In Section~\ref{subsec: method-refine}, we describe how to construct the policy when such prior knowledge is unavailable. 

Given the raw-data DAG in Figure~\ref{fig: reduce}a, we first refine it to retain only justifiably fair causal pathways, following Proposition~\ref{prop: fair}. Specifically, we remove all directed edges from sensitive and inadmissible attributes to the label, as shown in Figure~\ref{fig: reduce}b. This yields a fair policy in which sensitive attributes influence the label only through admissible attributes. We preserve this structure in subsequent steps, so the fairness guarantee holds throughout.

Next, we augment the policy by introducing a latent attribute to capture utility-relevant signals not fully represented in the observed attributes.
The key idea is to model hidden factors that affect both inadmissible attributes $\mathcal{I}$ and the label $Y$ while remaining independent of sensitive attributes. 
This design treats inadmissible attributes as jointly determined by observed sensitive factors and an unobserved latent factor, allowing $L$ to extract decision-relevant signals. Since $L$ is independent of the sensitive attributes by design, the fairness guarantees established above are preserved. 

\vspace{2mm}
\noindent\textbf{First-cut solution.} 
A straightforward approach is to model $L$ as a source node that affects all inadmissible attributes $\mathcal{I}$ and the label $Y$. However, connecting $L$ to the entire set $\mathcal{I}$ substantially increases model complexity and makes data calibration computationally expensive when $|\mathcal{I}|$ is large. We therefore adopt a more fine-grained design by allowing $L$ to affect only the label and a selected subset of inadmissible attributes, namely the direct parents of $Y$, denoted by $\mathcal{I}_c := \mathcal{I} \cap \Pi_Y$. The intuition is that, under the local Markov property and Markov compatibility, the effect of $\mathcal{I}$ on $Y$ is mainly captured through $\mathcal{I}_c$, as reflected in Equation~\eqref{eq: bayes} and supported by the analysis in Equation~\eqref{eq: Pgl'_approx} in the next subsection. Modeling $L$ as connecting only to $\mathcal{I}_c$ thus preserves sufficient decision-relevant information while keeping data calibration scalable. 

\vspace{2mm}
\noindent\textbf{Refinement for identifiability.} 
A subtle yet critical limitation of the first-cut solution is that it overlooks \emph{identifiability}, that is, the ability to uniquely infer the effect of the latent attribute from observed data. Without identifiability, $L$ may capture redundant or unstable signals. This issue becomes particularly severe when the attributes in $\mathcal{I}_c$ are tightly connected, as their information may then collapse into a single shared source of variation. In such cases, $L$ behaves as if it interacts with only one inadmissible attribute and the label, weakening the statistical basis for estimating $L$ and leading to ambiguous or meaningless contributions. 

To tackle this, we introduce a pruning mechanism that reshapes the local causal structure around $\mathcal{I}_c$ to improve identifiability. Since adopting a fairness policy already removes all direct causal effects from inadmissible attributes to the label, it prevents the information in $\mathcal{I}_c$ and $Y$ from collapsing into a single source. Building on this property, it suffices to further limit dependence within $\mathcal{I}_c$. Accordingly, \latent\ partitions $\mathcal{I}_c$ into two conditionally independent subsets $\mathcal{I}_{c_1}$ and $\mathcal{I}_{c_2}$, and eliminates direct interconnections between them. The latent attribute $L$ then points to both subsets and to the label $Y$, as illustrated in Figure~\ref{fig: reduce}c. 
Formal identifiability analysis and algorithmic details are provided in Section~\ref{sec: component}.

Beyond identifiability and implementability, this design also provides an alternative view of robustness. By distributing fairness-aware adjustments across the DAG, it avoids making any single region overly sensitive to imperfections in attribute specification. As shown in Figure~\ref{fig: reduce}b, fairness is enforced by severing all edges from $\mathcal{S} \cup \mathcal{I}$ to $Y$, localizing modifications near the label. In contrast, Figure~\ref{fig: reduce}c introduces the latent variable $L$ to mediate the influence of $\mathcal{I}_c$ on $Y$, preserving utility-relevant signals while adjusting local dependencies among $\{V_4, V_5\}$. This redistribution spreads the impact of fairness constraints more evenly, preventing structural collapse and maintaining overall stability and utility.

\subsection{DAG Refinement} \label{subsec: method-refine}

A key challenge in applying this policy is that the underlying DAG is rarely available in practice~\cite{salimi2019interventional, pirhadi2024otclean}. While causal discovery methods such as Max-Min Hill Climbing (MMHC)~\cite{tsamardinos2006max} can estimate the DAG structure, their computational cost grows exponentially with the number of attributes~\cite{tsamardinos2006max}, posing serious scalability concerns. This motivates a more efficient procedure for extracting only the structural information needed for policy enforcement. 

We address this challenge using a two-step refinement strategy inspired by CausalPre~\cite{zheng2025causalpre}, but tailored to support latent attribute integration. The key idea is to construct a coarse-grained fairness policy that avoids full DAG reconstruction while preserving causally fair information. 

\emph{In the first step,} we identify relationships that remain unchanged throughout processing and abstract them as a black box. Since these structures require no modification, the policy only needs to indicate that they are preserved, without specifying their internal form. 
As discussed in Section~\ref{subsec: method-framework} and formalized in Proposition~\ref{prop: fair}, fairness-aware pre-processing only modifies the structure around the label $Y$. To recover hidden utility-relevant signals, our framework further updates the region around $\mathcal{I}_c$. These localized adjustments confine the scope of modification to the neighborhoods of $Y$ and $\mathcal{I}_c$, leaving the rest intact. We therefore treat these preserved relationships among $\mathcal{V} \setminus (\mathcal{I}_c \cup {Y})$ collectively as a black box, yielding the simplification from Figure~\ref{fig: reduce}c to Figure~\ref{fig: reduce}d. In this representation, only the edges directed to $Y$ and $\mathcal{I}_c$ (shown as solid black arrows) are explicitly encoded in the fairness policy, while all others are abstracted away. Although the figure shows only a few attributes, in practice, most structural information is absorbed into the black-box region, substantially reducing computational cost.

\emph{In the second step,} we refine the structure around $Y$ and $\mathcal{I}_c$ by greedily expanding their parent sets, avoiding expensive exact parent computation. As shown in~\cite{zheng2025causalpre}, this does not compromise informational utility. For $Y$, we admit all attributes in $\mathcal{A}\cup\mathcal{W}$ as potential parents. For each $V_i \in \mathcal{I}_c$, its parents are drawn from $\mathcal{S}\cup\mathcal{I}\cup\mathcal{A}$, ensuring that sensitive information comes from observable sources and that the latent attribute reflects only unbiased signals. 

However, since each $V_i$ ($\in \mathcal{I}_c$) also belongs to $\mathcal{I}$, naive parent selection may introduce self-loops or cycles. To avoid this, we further abstract the relationships within $\mathcal{I}_c$. Let $\mathcal{I}_o := \mathcal{I} \setminus \mathcal{I}_c$ denote the remaining inadmissible attributes. Because dependencies within each partition of $\mathcal{I}_c$ remain locally unchanged during construction, we treat them as two additional black boxes. Each black box then admits $\mathcal{S} \cup \mathcal{I}_o \cup \mathcal{A}$ as its effective parent set, preventing cyclic dependencies while preserving necessary information flow. This yields the coarse-grained representation in Figure~\ref{fig: reduce}e. 

Although this refinement may introduce additional dependencies to $Y$ and $\mathcal{I}_c$, these effects are acceptable. For $Y$, the added decision factors are admissible by definition. For $\mathcal{I}_c$, its influence on $Y$ is already blocked, thus new dependencies do not affect decision-making or violate fairness. While this may slightly distort the local distribution around $\mathcal{I}_c$, it enables a more faithful reconstruction of the distribution around $Y$ through the latent variable $L$. Experimental results in Section~\ref{sec: exp} confirm that this trade-off yields clear performance gains. Moreover, such additional dependencies are rare in practice, as real-world datasets typically exhibit rich attribute interactions that already account for most dependencies. 

The approximated distribution under the refined coarse-grained DAG takes the general form:
\begin{align} \label{eq: Pgl'_approx}
    \mathbb{P} [L, \mathcal{V}]
   &\approx \mathbb{P}[\mathcal{V}\setminus(\mathcal{I}_c\cup\{Y\})] \cdot \mathbb{P}[L] \cdot \mathbb{P}[\mathcal{I}_{c_1}\mid \Pi_{c_1}'] \notag \\
   &\quad\quad\quad \cdot \mathbb{P}[\mathcal{I}_{c_2}\mid \Pi_{c_2}'] \cdot \mathbb{P}[Y\mid \Pi_Y'],
\end{align}
where $\Pi_{c_1}'=\{L,\mathcal{S},\mathcal{I}_o,\mathcal{A}\}$, $\Pi_{c_2}'=\{L,\mathcal{S},\mathcal{I}_o,\mathcal{A}\}$, and $\Pi_Y'=\{L,\mathcal{A},\mathcal{W}\}$ represent the greedy fair parent sets of $\mathcal{I}_{c_1}$, $\mathcal{I}_{c_2}$, and $Y$, respectively.

\subsection{The Complete Framework} \label{subsec: method-complete}

With the refined fairness policy in place, \latent\ adjusts the raw data to satisfy its structural constraints while preserving the original distribution as much as possible. This amounts to estimating the optimal parameters $\theta^*$ for the refined DAG that best approximate the raw data distribution. Details of the implementation are provided in Section~\ref{subsec: method-estimation}.

Algorithm~\ref{algo: complete} summarizes the full processing pipeline. Given a raw dataset $\mathcal{D}$, \latent\ first identifies the set $\mathcal{I}_c$ and partitions it to construct the latent-augmented fairness policy based on Equation~\eqref{eq: Pgl'_approx}. It then estimates the optimal parameter set $\theta^*$ under this policy and treats the induced distribution as the target for data processing. Given $\theta^*$, \latent\ samples each attribute set ``$i$'' from its factor distribution $\mathbb{P}_{\theta_i}$ and overwrites the corresponding columns in the dataset. After marginalizing out the latent attribute $L$, it outputs the processed dataset $\mathcal{D}'$.

\begin{algorithm}[!t]
    \caption{\textsc{DataPreprocessing}} 
    \label{algo: complete}
        \KwIn{Database $\mathcal{D}$ with attribute set $\mathcal{V}=\mathcal{S}\cup\mathcal{I}\cup\mathcal{A}\cup\mathcal{W}\cup\{Y\}$, the number of latent states $\tau$, maximum number of iterations $n$, convergence threshold $\eta$, maximum size of conditioning set $\alpha$}
        \KwOut{Processed data $\mathcal{D}'$} 
        
        $\mathcal{I}_{c} \gets$ \textsc{Identification}($\mathcal{D}$, $\alpha$)\tcp*{we set $\alpha = 2$}
        $\mathcal{I}_{o} \gets \mathcal{I} \setminus \mathcal{I}_c$\;
        $\mathcal{I}_{c_1}, \mathcal{I}_{c_2} \gets$ \textsc{Partition}($\mathcal{D}, \mathcal{I}_{c}, \mathcal{S}\cup\mathcal{I}_o\cup\mathcal{A}, \tau$)\;

        $\theta^* \gets$ \textsc{ParameterEstimation}($\mathcal{D}, \mathcal{I}_{c_1}, \mathcal{I}_{c_2}, n, \eta$ )\;

        $k \gets \vert \mathcal{D} \vert$\;
        
        \For{$\theta_i \in \theta^*$}{ 
            Sample $k$ values for attributes in set ``$i$'' from distribution $\mathbb{P}_{\theta_i}$, and fill the corresponding columns in $\mathcal{D}'$\; \label{line: sample k}
        } 

        Marginalize out latent attribute $L$ from $\mathcal{D}'$\;

       \Return{$\mathcal{D}'$}\;
\end{algorithm}


\section{Algorithmic Design and Analysis} \label{sec: component}

This section presents the theoretical analysis and algorithmic design of \latent. We first establish the identifiability of the latent-augmented DAG in Section~\ref{subsec: method-identifiability}. We then describe the algorithms for identifying and partitioning the attribute set $\mathcal{I}_c$ in Sections~\ref{subsec: method-identification} and~\ref{subsec: method-partition}. Finally, in Section~\ref{subsec: method-estimation}, we describe the parameter estimation procedure for enforcing the policy and obtaining the optimal parameters $\theta^*$ for data adjustment.

\subsection{Identifiability Analysis} \label{subsec: method-identifiability}

To enable formal identifiability analysis, we simplify the DAG by grouping related attributes into composite units~\cite{zeng2025causal, loukas2019graph, hashemi2024comprehensive} and treating each composite as a single clumped attribute~\cite{allman2009identifiability}. The resulting generalized model is shown in Figure~\ref{fig: dag_dual}, where $X_0$ denotes the latent attribute $L$, $X_1$ and $X_2$ represent the two disjoint subsets $\mathcal{I}_{c_1}$ and $\mathcal{I}_{c_2}$, $X_3$ denotes the label $Y$, and $X_4$ aggregates $\mathcal{S}\cup\mathcal{I}_o\cup\mathcal{A}$. The following theorem shows that, under mild conditions, this generalized model is identifiable, from which the identifiability of \latent\ follows directly. 
\begin{theorem} \label{thm: dual}
    Consider a latent causal model with one latent attribute $X_0$ and $d \geq 4$ observed attributes. The observed attributes are partitioned into four nonempty groups $X_1, X_2, X_3, X_4$, such that conditional on $(X_0,X_4)$ the groups $X_1, X_2, X_3$ are mutually independent, and moreover $X_0 \perp X_4$. The latent attribute $X_0$ takes $\tau \ge 2$ states with strictly positive mixture weights. For each group $X_i$ with $i\in \{1,2,3,4\}$, define its group cardinality as \[ \kappa_i := \prod_{j\in X_i} k_j, \] where $k_j$ is the number of states of the attribute $j$.
    
    Assume $\kappa_i \geq 2$ for all $i$. If $\tau$ does not exceed the cardinality of at least two of the groups among $X_1, X_2, X_3$, that is, if 
    \begin{align} \label{eq: latent_limit}
        \tau \;\leq\; \min \bigl\{ \max(\kappa_1,\kappa_2),\ \max(\kappa_2,\kappa_3),\ \max(\kappa_1,\kappa_3) \bigr\}, 
    \end{align}
    then the model parameters are generically identifiable up to a permutation of the latent states. 
\end{theorem}
\begin{proof}
    Let $\theta_i$ denote the local parameter at $X_i$ for $i\in\{0,1,2,3,4\}$. We aim to show that the local parameters $\theta_i$ of each $X_i$ are generically identifiable up to a global permutation of the latent states.
    
    Fix any value $j \in \{1,\ldots,\kappa_4\}$ of $X_4$. Conditioning on $X_4=j$, we obtain a latent class model with three observed attribute groups $(X_1, X_2, X_3)$ and one latent variable $X_0$. For $i=0,1,2,3$, denote these conditioned attributes or groups by $X_i^{(j)}$ and let $\theta_i^{(j)}$ denote their associated local parameters.

    By assumption on the group cardinalities, at least two of $\kappa_1,\kappa_2,\kappa_3$ are not smaller than $\tau$. Without loss of generality, assume $\kappa_1,\kappa_2 \geq \tau$. Then, for each conditional model $(X_0^{(j)}, X_1^{(j)}, X_2^{(j)}, X_3^{(j)})$, the identifiability condition in Equation~\eqref{eq: latent_limit}, derived from tensor decomposition~\cite{allman2009identifiability, allman2015parameter}, is satisfied. It therefore follows that the parameters $(\theta_0^{(j)}, \theta_1^{(j)}, \theta_2^{(j)}, \theta_3^{(j)})$ are generically identifiable, up to a permutation $\pi_j$ of the latent states.
    
    Since the mixing distribution of $X_0$ does not depend on $X_4$, we have $\theta_0^{(j)} = \theta_0$ for all $j$. This implies that the permutations $\pi_j$ obtained for different $j$ must be consistent; otherwise, they would yield distinct versions of $\theta_0$, contradicting uniqueness. Therefore, all $\pi_j$ coincide with a single global permutation $\pi$. Consequently, $(\theta_0, \theta_1, \theta_2, \theta_3)$ are generically identifiable up to a global permutation of the latent states. Finally, the parameters $\theta_4$ are directly estimated from the conditional distributions $\mathbb{P} \left[ X_4 \right]$, which are part of the observed distribution. This completes the proof.
\end{proof}
In addition to identifiability, Theorem~\ref{thm: dual} provides a practical upper bound on the number of latent states. Since the label domain is often small, we conservatively set
\begin{align} \label{eq: latent_limit_simple}
    \tau \leq \min {(|\mathcal{I}_{c_1}|, |\mathcal{I}_{c_2}|)}.
\end{align}

\subsection{Identification of $\mathcal{I}_c$} \label{subsec: method-identification}

To construct the latent-augmented fairness policy, we must identify the attribute set $\mathcal{I}_c$, as mentioned in Section~\ref{subsec: method-framework}. A natural approach is to first recover the parent set $\Pi_Y$ of the label $Y$ and then compute $\mathcal{I}_c = \mathcal{I} \cap \Pi_Y$. However, this requires $O(d \cdot 2^d)$ conditional independence (CI) tests, since for each attribute $V_i \in \mathcal{V} \setminus {Y}$, we must evaluate whether $V_i \perp Y \mid \Omega$ for all subsets $\Omega \subseteq \mathcal{V} \setminus \{V_i, Y\}$. This exponential cost is prohibitive in practice.

Given that the inadmissible set is typically much smaller than the full attribute set, we instead test each inadmissible attribute individually to determine whether it is a parent of $Y$. This reduces the number of CI tests to $O(|\mathcal{I}| \cdot 2^d)$. Although still exponential, empirical studies~\cite{kocaoglu2023characterization, sondhi2019reduced} show that limiting the size of conditioning sets to 2 or 3 is usually sufficient in practice, which reduces the overall cost to $O(|\mathcal{I}| \cdot d)$. In this paper, we set the maximum size of the conditioning set, denoted by $\alpha$, to $2$ by default.

Algorithm~\ref{algo: identification} outlines the identification procedure, which consists of two steps: (i) test whether a pathway exists from $X$ to $Y$, and (ii) if so, determine whether it is direct. In the initial step (Lines~\ref{line: chi2 start}-\ref{line: chi2 end}), we apply the chi-square test to evaluate whether an inadmissible attribute $X$ is independent of $Y$. If independence is confirmed, we eliminate $X$ from further consideration. Otherwise, we proceed to the second step (Lines~\ref{line: g start}-\ref{line: g end}), using the G-test to evaluate whether the dependence persists after conditioning on various subsets $Z$. If the dependency holds across all such $Z$, we conclude that no set blocks the path, indicating a direct causal link. This identification process is computationally efficient due to the linear number of tests and the early-stop mechanism.

\begin{algorithm}[!t]
    \caption{\textsc{Identification}} 
    \label{algo: identification}
        \KwIn{Database $\mathcal{D}$ with attribute set $\mathcal{V} = \mathcal{S} \cup \mathcal{I} \cup \mathcal{A} \cup \mathcal{W} \cup \{Y\}$, maximum size of conditioning set $\alpha$}
        \KwOut{Inadmissible subset $\mathcal{I}_{c}$}
        
        Let $\mathcal{I}_{c} \gets \mathcal{I}$ \;
        
        \For{$m \gets 0 $ to $ \alpha$}{
            \For{each attribute $X \in \mathcal{I}_{c}$}{
                \eIf{$m = 0$}{ 
                    Test if $X \perp Y$ using chi-square test\; \label{line: chi2 start}
                    If independent: Remove $X$ from $\mathcal{I}_{c}$\; \label{line: chi2 end}
                }{
                    $\mathcal{Z} \gets $ Construct conditioning sets of size $m$ from $\mathcal{V} \setminus\{X, Y\}$\; \label{line: g start}
                    Test if $X \perp Y \mid Z$ using G-test, where $Z \in \mathcal{Z}$\;
                    If independent for any $Z$: Remove $X$ from $\mathcal{I}_{c}$ and break\; \label{line: g end}
                }
            }
        }
        
        \Return{$\mathcal{I}_{c}$}\;
\end{algorithm}

\subsection{Partition of $\mathcal{I}_c$} \label{subsec: method-partition}

After identifying the set $\mathcal{I}_c$, we partition it into two disjoint and non-empty subsets, $\mathcal{I}_{c_1}$ and $\mathcal{I}_{c_2}$, to enable reasonable integration of the latent attribute. The subsets $\mathcal{I}_{c_1}$ and $\mathcal{I}_{c_2}$ are required to be conditionally independent given $\mathcal{Z}=\mathcal{S}\cup\mathcal{I}_o\cup\mathcal{A}$ to satisfy the identifiability condition described in Section~\ref{subsec: method-identifiability}. To preserve as much information as possible, we formulate the partitioning process as a constrained optimization problem that aims to find two sufficiently large subsets with minimal dependence under the given context:
\begin{align*}
    \min_{\mathcal{I}_{c_1}, \mathcal{I}_{c_2}} \; I(\mathcal{I}_{c_1}; \mathcal{I}_{c_2} \mid \mathcal{Z}) \quad \text{s.t.} \quad \vert \mathcal{I}_{c_1} \vert \geq \tau, \vert \mathcal{I}_{c_2} \vert \geq \tau,
\end{align*}
where $\tau$ represents the number of latent states and the constraint ensures that each subset is large enough for the latent attribute to represent meaningful states, as bounded by Equation~\eqref{eq: latent_limit_simple}.

This optimization is computationally intractable due to two challenges: computing multivariate conditional mutual information (MCMI) is costly, and the number of possible bipartitions grows exponentially with the size of $\mathcal{I}_c$. To address this, we apply two approximations. First, we estimate the MCMI using the sum of pairwise conditional dependencies (i.e., CMI) ~\cite{chow1968approximating} as shown below:
\begin{align} \label{eq: approx-cmi}
    I(\mathcal{I}_{c_1}; \mathcal{I}_{c_2} \mid \mathcal{Z})
    \approx \sum_{X \in \mathcal{I}_{c_1}} \;\sum_{Y \in \mathcal{I}_{c_2}} I(X;Y \mid \mathcal{Z}).
\end{align}
This captures the essential cross-set dependencies while remaining computationally efficient. 
Second, we employ a greedy hill-climbing strategy to search for a near-optimal partition, as described in Algorithm~\ref{algo: partition}. The algorithm begins with a random bipartition of $\mathcal{I}_c$ that satisfies the cardinality constraint (Lines~\ref{line: partition-init-start}--\ref{line: partition-init-end}), then iteratively refines it through two steps. In the first step (1-move), it considers moving a single attribute from one group to the other and selects the move that most reduces the estimated CMI (Lines~\ref{line: partition-move-start}--\ref{line: partition-move-end}). If no such move is valid, it performs the second step (2-swap) by exchanging one attribute from each group (Lines~\ref{line: partition-swap-start}–\ref{line: partition-swap-end}). This exchange step is particularly helpful when the partition is close to the cardinality threshold and single moves would violate the constraint. The algorithm repeats these updates until no significant CMI reduction is observed, as determined by the relative tolerance parameter $\varepsilon$. The parameter $\varepsilon$ controls the trade-off between runtime and partition quality, and we set $\varepsilon=10^{-5}$ in experiments to ensure high-quality partitions with reasonable runtime. The final output is a partition $(\mathcal{I}_{c_1}, \mathcal{I}_{c_2})$, which serves as the basis for policy construction. 

\vspace{2mm}
\noindent\textbf{Complexity analysis.} Let $m=|\mathcal{I}_c|$. As described, the while loop updates the partition $(\mathcal{I}_{c_1}, \mathcal{I}_{c_2})$ only when the CMI improvement satisfies $\Delta^* \ge \varepsilon n$ and stops once $n \le \varepsilon n_0$. Since each accepted update reduces the estimated CMI value by at least a multiplicative factor of $(1-\varepsilon)$, the loop performs at most $O(\log(\varepsilon)/\log(1-\varepsilon))=O(\frac{1}{\varepsilon} \log\frac{1}{\varepsilon})$ iterations. In each iteration, the dominant cost comes from the 2-swap step, which considers at most $|\mathcal{I}_{c_1}||\mathcal{I}_{c_2}| \le m^2/4$ candidate swaps. For each candidate partition, evaluating Equation~\eqref{eq: approx-cmi} involves $O(m^2)$ pairwise CMI evaluations. Therefore, each iteration requires at most $O(m^4)$ pairwise CMI evaluations, and the overall procedure requires at most $O(\frac{1}{\varepsilon}\log\frac{1}{\varepsilon} \cdot m^4)$ such evaluations.

\subsection{Parameter Estimation} \label{subsec: method-estimation}

Once the fairness policy is established, we adjust the raw data to satisfy its structural constraints while minimizing distortion to the underlying distribution, thereby completing the fairness-aware pre-processing.
Equivalently, we seek the parameter setting $\theta^*$ of the refined DAG structure that is most faithful to the original data distribution among all feasible configurations. This can be cast as a maximum likelihood estimation (MLE) problem:
\begin{align}
    \theta^* = \text{arg}\max_\theta\ \log{\sum_{L}\mathbb{P}\left[\mathcal{V}, L\mid\theta\right]}. \label{eq: mle}
\end{align}
Because the latent variable $L$ appears inside the logarithm, direct maximization is intractable. 
To address this, we introduce an auxiliary distribution $\mathbb{P}\left[L\mid \mathcal{V},\hat{\theta}\right]$ with fixed parameter $\hat{\theta}$ and apply Jensen's inequality to derive a tractable lower bound:
\begin{align}
    \log \sum_{L}\mathbb{P}&\left[\mathcal{V}, L\mid\theta\right]
    = \log{\mathbb{E}_{L\sim \mathbb{P}\left[L\mid \mathcal{V},\hat{\theta}\right]}\left[\frac{\mathbb{P}\left[\mathcal{V}, L\mid\theta\right]}{\mathbb{P}\left[L\mid \mathcal{V},\hat{\theta}\right]}\right]} \notag \\
    &\geq \mathbb{E}_{L\sim \mathbb{P}\left[L\mid \mathcal{V},\hat{\theta}\right]}\bigg[\log\mathbb{P}\left[\mathcal{V}, L\mid\theta\right]\bigg] + H\bigg[\mathbb{P}\left[L\mid \mathcal{V},\hat{\theta}\right]\bigg], \label{eq: target1} 
\end{align}
where $H[\cdot]$ denotes entropy.
Since the second term is constant with respect to $\theta$, maximizing the lower bound reduces to
\begin{align}
    \theta^* = \text{arg}\max_\theta\ Q\left(\theta\mid\hat{\theta}\right), \label{eq: target3}
\end{align}
where
\begin{align*}
    Q\left(\theta\mid\hat{\theta}\right) \;&:=\; \mathbb{E}_{L\sim \mathbb{P}\left[L\mid \mathcal{V},\hat{\theta}\right]}\bigg[\log\mathbb{P}\left[\mathcal{V}, L\mid\theta\right]\bigg].
\end{align*}
The validity and convergence of this reformulation follow the standard analysis in~\cite{dempster1977maximum, wu1983convergence}. 

If $L$ were observed, Equation~\eqref{eq: target3} could be optimized directly. Since $L$ is latent, we treat it as a column of missing data and adopt the Expectation-Maximization (EM) paradigm to jointly estimate both $L$ and $\theta$. The procedure alternates between two steps. In the E-step, we consider the current parameter $\hat{\theta}$ to be fixed and estimate the distribution of $L$, which provides a probabilistic imputation of the missing values. In the M-step, we treat the data as complete and optimize $Q\left( \theta \mid \hat{\theta} \right)$ to update $\theta$. This process repeats until convergence, yielding the final estimate $\theta^* = \theta^{(T)}$ after $T$ iterations. 
Both steps are tailored to exploit the structural properties of our refined fairness policy, which is key to achieving efficient computation.

\begin{algorithm}[!t]
    \caption{\textsc{Partition}} 
    \label{algo: partition}
        \KwIn{Database $\mathcal{D}$, inadmissible subset $\mathcal{I}_c$, conditioning set $\mathcal{Z}$, the number of latent states $\tau$, relative tolerance $\varepsilon$}
        \KwOut{Subsets $\mathcal{I}_{c_1}$ and $\mathcal{I}_{c_2}$}
        
        Initialize $\mathcal{I}_{c_1}, \mathcal{I}_{c_2} \gets $ random bipartition of $\mathcal{I}_c$\; \label{line: partition-init-start}
        Check feasibility: if $|\mathcal{I}_{c_1}| < \tau$ or $|\mathcal{I}_{c_2}| < \tau$, repeatedly move a randomly selected attribute $v$ from the larger subset to the smaller one until both subsets satisfy the cardinality threshold\; \label{line: partition-init-end}
        
        $n \leftarrow I(\mathcal{I}_{c_1}; \mathcal{I}_{c_2} \mid \mathcal{Z})$ (as shown in Equation~\eqref{eq: approx-cmi}),\!\!\quad $n_0\gets n$\;

        \While{true}{
            Initialize $\Delta^\star \leftarrow 0$,\!\!\quad $(\mathrm{A}^\star,\mathrm{B}^\star)\gets\bot$\;

            \tcp{(1) Apply the best 1-move}
            \For{$v \in \mathcal{I}_c$}{ \label{line: partition-move-start}
                Form candidate partition $(\mathrm{A},\mathrm{B})$ by moving $v$ to the opposite partition\;
                \If{$|\mathrm{A}|\geq \tau, |\mathrm{B}| \geq \tau$}{
                    Compute $n' \leftarrow I(\mathrm{A};\mathrm{B} \mid \mathcal{Z})$,\!\!\quad $\Delta \leftarrow n-n'$\;
                    \If{$\Delta>\Delta^\star$}{Update $\Delta^\star,(\mathrm{A}^\star,\mathrm{B}^\star)$\;} \label{line: partition-move-end}
                }
            }

            \BlankLine
            \tcp{(2) No improving 1-move: try the best 2-swap}
            \If{$\Delta^\star = 0$}{ \label{line: partition-swap-start}
                \For{$(u,v) \in \mathcal{I}_{c_1} \times \mathcal{I}_{c_2}$}{
                    Form candidate partition $(\mathrm{A},\mathrm{B})$ by moving $u$ and $v$ to the opposite partition\;
                    \If{$|\mathrm{A}| \geq \tau, |\mathrm{B}| \geq \tau$}{
                        Compute $n' \leftarrow I(\mathrm{A};\mathrm{B} \mid \mathcal{Z})$,\!\!\quad $\Delta \leftarrow n-n'$\;
                        \If{$\Delta>\Delta^\star$}{Update $\Delta^\star,(\mathrm{A}^\star,\mathrm{B}^\star)$\;} \label{line: partition-swap-end}
                    }
                }
            }
            \If{$\Delta^\star < \varepsilon  n$ {\rm or} $n \le \varepsilon  n_0$}{\textbf{break}}
            Update $(\mathcal{I}_{c_1}, \mathcal{I}_{c_2})\leftarrow(\mathrm{A}^\star,\mathrm{B}^\star)$,\!\!\quad $n\leftarrow n-\Delta^\star$\;
        }
        \Return $\mathcal{I}_{c_1}, \mathcal{I}_{c_2}$
\end{algorithm}

\subsubsection{Analysis of E-step}

At $k$-th iteration, parameters are fixed as $\hat{\theta}^{(k)}$, initialized randomly for $k{=}1$ and updated as $\hat{\theta}^{(k)}:=\theta^{(k-1)}$ if $k{>}1$. The goal is to infer the posterior of $L$ given $\mathcal{V}$:
\begin{align*}
    \mathbb{P}\left[L\mid \mathcal{V}; \hat{\theta}^{(k)}\right] 
    = \frac{\mathbb{P}\left[\mathcal{V}, L; \hat{\theta}^{(k)}\right]}{\mathbb{P}\left[\mathcal{V}; \hat{\theta}^{(k)}\right]} 
    \propto \mathbb{P}\left[\mathcal{V}, L; \hat{\theta}^{(k)}\right].
\end{align*}
By factorizing the refined DAG, as shown in Equation~\eqref{eq: Pgl'_approx}, we obtain: 
\begin{align}
    \mathbb{P}\left[L\mid \mathcal{V}; \hat{\theta}^{(k)}\right] 
    &\propto\ \mathbb{P}\left[L; \hat{\theta}^{(k)}\right] \cdot \mathbb{P}\left[\mathcal{I}_{c_1}\mid \Pi_{c_1}'; \hat{\theta}^{(k)}\right] \notag \\ 
    &\quad \cdot \mathbb{P}\left[\mathcal{I}_{c_2}\mid \Pi_{c_2}'; \hat{\theta}^{(k)}\right] \cdot \mathbb{P}\left[Y\mid \Pi_Y'; \hat{\theta}^{(k)}\right]. \label{eq: E-target}
\end{align}
Since only four terms involve the latent attribute $L$, we compute their product and normalize to obtain the posterior, instead of computing the full joint distribution. This yields a soft assignment of $L$ for each record, effectively completing the data over $\mathcal{V}\cup{L}$ for the subsequent M-step.

\subsubsection{Analysis of M-step} 

In this step, we update parameters using the posterior $\mathbb{P}\left[L\mid \mathcal{V}; \hat{\theta}^{(k)}\right]$ from the E-step. Let \( \omega_j(l) := \mathbb{P} \left[ L=l \mid \mathcal{V}^j; \hat{\theta}^{(k)} \right] \) denote the posterior weight of latent state $l$ for record $j$. The objective is to maximize:
\begin{align}
    \theta^{(k)} &= \text{arg}\max_\theta\ \sum_{j}{\sum_{l}{\bigg( \omega_j(l) \cdot \log\mathbb{P}\left[\mathcal{V}^j, L=l; \theta\right] \bigg)}}. \label{eq: M1}
\end{align}
Using the DAG factorization in Equation~\eqref{eq: Pgl'_approx}, 
\begin{align*}
    \mathbb{P}\left[\mathcal{V}, L; \theta\right] \propto \mathbb{P}[L; \theta] \cdot \mathbb{P}[\mathcal{I}_{c_1}\mid \Pi_{c_1}'; \theta] \cdot \mathbb{P}[\mathcal{I}_{c_2}\mid \Pi_{c_2}'; \theta] \cdot \mathbb{P}[Y\mid \Pi_Y'; \theta],
\end{align*}
the objective decomposes additively across local conditionals. In the discrete setting, each local conditional is multinomial, so parameters admit EM closed-form updates as normalized expected counts. 
Let $N_E(\cdot)$ denote expectations under the posterior. The updates are:
\begin{align}
    \theta_l = \frac{N_E(l)}{\sum_{l}{N_E(l)}}, \quad\quad
    &\theta_{i_{c_1} \mid \pi_{c_1}} = \frac{N_E(i_{c_1}, \pi_{c_1})}{\sum_{i_{c_1}}{N_E(i_{c_1}, \pi_{c_1})}}, \notag \\
    \theta_{y \mid \pi_y} = \frac{N_E(y, \pi_y)}{\sum_{y}{N_E(y, \pi_y)}}, \quad\quad
    &\theta_{i_{c_2} \mid \pi_{c_2}} = \frac{N_E(i_{c_2}, \pi_{c_2})}{\sum_{i_{c_2}}{N_E(i_{c_2}, \pi_{c_2})}}, \label{eq: p_local}
\end{align}
where \[ N_E(y, \pi_y) = \sum_{j}\sum_{l} \bigg( \mathbb{P}\left[L=l\mid \mathcal{V}^j; \hat{\theta}^{(k)}\right] \cdot \mathbb{I}\left[Y^j = y, \Pi_Y'^j=\pi_y \right] \bigg), \] and analogously for other factors. 
The final $\theta^{(k)}$ is formed by concatenating all local parameters.

\subsubsection{Overall.}

\begin{algorithm}[!t]
    \caption{\textsc{ParameterEstimation}} 
    \label{algo: estimation}
        \KwIn{Database $\mathcal{D}$ with attribute set $\mathcal{V}=\mathcal{X}\cup\mathcal{I}_{o}\cup\mathcal{I}_{c_1}\cup\mathcal{I}_{c_2}\cup\mathcal{A}\cup\mathcal{W}\cup\{Y\}$, maximum number of iterations $n$, convergence threshold $\eta$}
        \KwOut{Parameters $\theta^*$} 

        Initialize $\theta_l, \theta_{i_{c_1} \mid \pi_{c_1}}, \theta_{i_{c_2} \mid \pi_{c_2}}, \theta_{y \mid \pi_y}$ randomly\;
        
        \For{$i \gets 0 \text{ to } n$}{ \label{line: em start}
            $\mathbb{P}_{L|\mathcal{V}} \gets$ Equation~\eqref{eq: E-target}\;
            Update $\theta_l, \theta_{i_{c_1} \mid \pi_{c_1}}, \theta_{i_{c_2} \mid \pi_{c_2}}, \theta_{y \mid \pi_y}$ based on Equation~\eqref{eq: p_local}\;
            Check log-likelihood according to Equation~\eqref{eq: l}: if the change is less than $\eta$, stop the iteration\;
            $i \gets i+1$
        } \label{line: em end}

       Compute $\theta^*$ based on Equations~\eqref{eq: p_x} and~\eqref{eq: p_all}\;
        \Return{$\theta^*$}\;
\end{algorithm}

After each M-step, we calculate the refined log-likelihood as
\begin{align}
    \mathcal{L}^{(k)} \propto \sum_{j}
    \log\!\left( \sum_{l} \theta_{l}^{(k)} \cdot \theta_{\,i_{c_1}^j \mid \pi_{c_1}^{\prime\,j}(l)}^{(k)} \cdot \theta_{\,i_{c_2}^j \mid \pi_{c_2}^{\prime\,j}(l)}^{(k)} \cdot \theta_{\,y^j \mid \pi_{y}^{\prime\,j}(l)}^{(k)} \right) \label{eq: l}
\end{align}
and terminate the procedure when $\vert \mathcal{L}^{(k)} - \mathcal{L}^{(k-1)} \vert \leq \eta$ or the iteration count reaches $n$. Once converged, we obtain the optimal parameters for the local structures around $L$, $\mathcal{I}_{c_1}$, $\mathcal{I}_{c_2}$, and $Y$. For the remaining attributes, denoted as $\mathcal{X} = \mathcal{V}\setminus(\mathcal{I}_c\cup\{Y\})$, the parameters are constant and can be computed directly via MLE:
\begin{align}
    \theta_x = \frac{N(x)}{\sum_{x'} N(x')}. \label{eq: p_x}
\end{align}
The final parameter set is formed by concatenating the estimates over all local distributions:
\begin{align}
    \theta^* = (\theta_x, \theta_l, \theta_{i_{c_1} \mid \pi_{c_1}}, \theta_{i_{c_2} \mid \pi_{c_2}}, \theta_{y \mid \pi_y}). \label{eq: p_all}
\end{align}
Note that high-dimensional distributions can be approximated using techniques in~\cite{chow1968approximating}; since this is not our focus, we treat all required distributions as given. The full parameter estimation procedure is summarized in Algorithm~\ref{algo: estimation}.

\section{Experiments} \label{sec: exp}

In this section, we evaluate the effectiveness of \latent\ through extensive experiments. We first describe the experimental setup in Section~\ref{subsec: exp-setting}. We then compare \latent\ with baselines by reporting end-to-end performance in Section~\ref{subsec: exp-endtoend} and visualizing missing-relationship recovery in Section~\ref{subsec: exp-recovery}. Finally, we provide a deeper analysis of \latent\ by testing the dependence between the learned latent attribute and sensitive attributes in Section~\ref{subsec: exp-indep}, examining how key parameters affect performance in Section~\ref{subsec: exp-param}, and profiling runtime in Section~\ref{subsec: exp-time}.

\subsection{Setup} \label{subsec: exp-setting}
We implement the proposed \latent\ framework in Python\footnote{Available at \url{https://github.com/iamzhengying/LatentPre.git}}. All experiments are conducted on a machine with two Xeon(R) Gold 6326@2.90 GHz CPUs and 256GB of DRAM.

\vspace{2mm}
\noindent\textbf{Datasets.}
We evaluate the performance of \latent\ on three real-world datasets, including Adult~\cite{adultData2024}, COMPAS~\cite{compasData2024}, and Census-KDD~\cite{censuskddData2024}, and two synthetic datasets generated following~\cite{markakis2024press}. The key statistics of these datasets are summarized in Table~\ref{tab: data}.

\begin{table}[t]
    \centering
    \caption{Dataset statistics.}
    \label{tab: data}
    \begin{small}
    \begin{tabular}{lccc}
        \toprule
        \textbf{Dataset} & \textbf{\#Records} & \textbf{\#Attributes} & \textbf{Avg. Dom} \\
        \midrule
        Adult       & 32,561    & 13    & 13.33 \\
        COMPAS      & 6,130     & 8     & 4     \\
        Census-KDD  & 196,130   & 28    & 11.67 \\
        Synthetic (Section~\ref{subsec: exp-recovery})   & 50,000    & 7     & 4     \\
        Synthetic (Section~\ref{subsec: exp-time})   & 50,000--500,000    & 15--60     & 4     \\
        \bottomrule
    \end{tabular}
    \end{small}
\end{table}

\vspace{2mm}
\noindent\textbf{Baselines.}
We compare \latent\ with state-of-the-art pre-processing methods, Cap-MS~\cite{salimi2019interventional}, Cap-MF~\cite{salimi2019interventional}, OTClean~\cite{pirhadi2024otclean}, and CausalPre~\cite{zheng2025causalpre}. For all of them, we use the source code provided by the authors. Since Cap-MS and Cap-MF are designed only for saturated CI constraints over all attributes, we extend them by first processing $\mathcal{D}_{\mathcal{V}\setminus\mathcal{W}}$ and then sampling $\mathcal{W}$ from its conditional distribution to complete the entire $\mathcal{D}'$. For OTClean~\cite{pirhadi2024otclean}, we report two variants to ensure a fair comparison, following the evaluation protocol in~\cite{zheng2025causalpre}: (i) OTClean-RT, which directly uses the provided codebase and thus modifies the testing data, and (ii) OTClean, which processes only the training data and leaves the testing set unchanged. We also report results on the original dataset, denoted as ``Original'', and the dataset with all sensitive and inadmissible attributes dropped, denoted as ``Dropped''. The former provides a baseline for understanding the dataset’s utility and inherent discrimination without any intervention, while the latter reflects what a naive brute-force approach can achieve.

\vspace{2mm}
\noindent\textbf{Measurement and metrics.}
We evaluate each framework by measuring the quality of its processed data along two dimensions: utility and fairness. To achieve this, each raw dataset is divided into a training set and a testing set; only the training set is pre-processed, while the testing set remains unchanged to simulate a realistic environment. Two standard classifiers, random forest (RF) and multilayer perceptron (MLP), are then trained on the processed training data and evaluated on the testing data. The resulting predictive performance then serves as a proxy for data quality, as higher-quality training data enables models to produce more accurate and less biased predictions. For utility, we report the AUC score, where a higher value indicates better generalization and thus higher data quality. For fairness, we use the Ratio of Observational Discrimination (ROD)~\cite{salimi2019interventional}, which measures deviation from fairness; formally 
\begin{align*}
    \text{ROD} = \max_{\text{\scriptsize$\mathcal{S}$}_0, \text{\scriptsize$\mathcal{S}$}_1 \in Dom(\mathcal{S})} \frac{1}{|Dom(\mathcal{A})|} \sum_{a \in Dom(\mathcal{A})} \overline{\text{ROD}} (\text{\scriptsize$\mathcal{S}$}_0, \text{\scriptsize$\mathcal{S}$}_1; \hat{Y} \mid a),
\end{align*}
where 
\begin{align*}
    \overline{\text{ROD}} (\text{\scriptsize$\mathcal{S}$}_0, \text{\scriptsize$\mathcal{S}$}_1; \hat{Y} \mid a) 
    = \frac{\mathbb{P}[\hat{Y}=1 \mid \text{\scriptsize$\mathcal{S}$}_0, a] \cdot \mathbb{P}[\hat{Y}=0 \mid \text{\scriptsize$\mathcal{S}$}_1, a]}{\mathbb{P}[\hat{Y}=0 \mid \text{\scriptsize$\mathcal{S}$}_0, a] \cdot \mathbb{P}[\hat{Y}=1 \mid \text{\scriptsize$\mathcal{S}$}_1, a]}.
\end{align*}
For consistent interpretation, we report the normalized absolute logarithm of the ROD value; a value of $0$ indicates no discrimination, and larger values indicate greater discrimination. Finally, to ensure robustness and reduce variability, all evaluations are conducted using five-fold cross-validation.

\begin{figure*}[htbp]
    \centering
    \hspace{2mm} \includegraphics[width=0.9\linewidth]{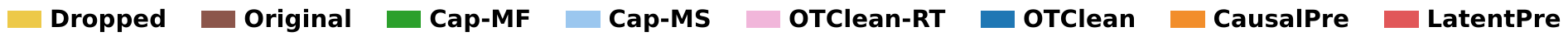}
    \vspace{2mm}

    \begin{minipage}[t]{\textwidth}
        \centering
        \begin{minipage}[t]{0.018\linewidth}
            \centering
            \hspace{0mm}\includegraphics[width=\linewidth, trim=0mm 0mm 0.4mm 0mm, clip]{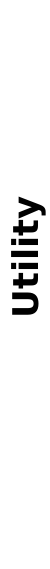}
            \vspace{-2mm}
        \end{minipage}
        \begin{minipage}[t]{0.32\linewidth}
            \centering
            \hspace{-2mm}\includegraphics[width=\linewidth]{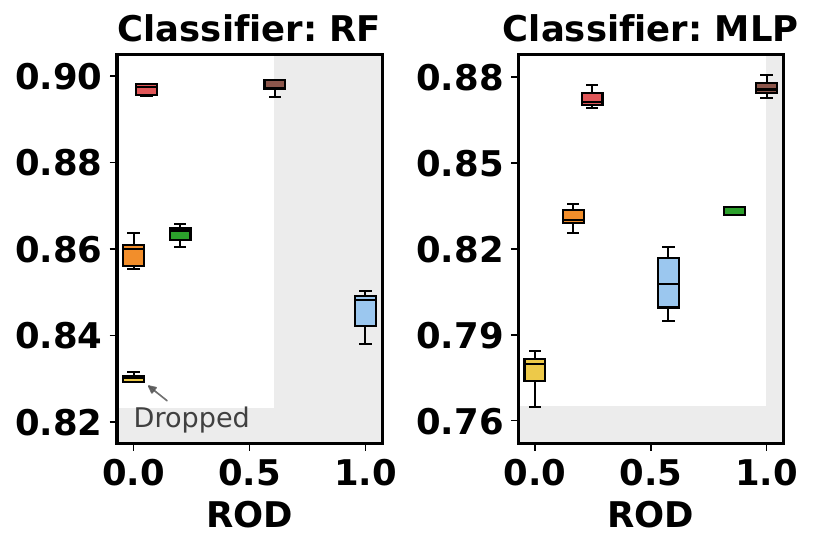}
            \vspace{-2mm} {\small (a) $\mathsf{Adult}$}
        \end{minipage}
        \begin{minipage}[t]{0.32\linewidth}
            \centering
            \hspace{-2mm}\includegraphics[width=\linewidth]{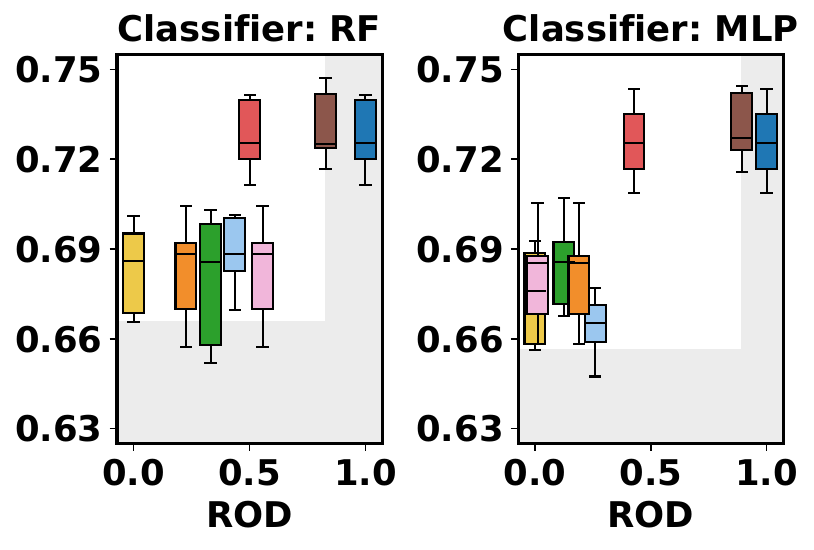}
            \vspace*{-2mm} {\small (b) $\mathsf{COMPAS}$}
        \end{minipage}
        \begin{minipage}[t]{0.32\linewidth}
            \centering
            \hspace{-2mm}\includegraphics[width=\linewidth]{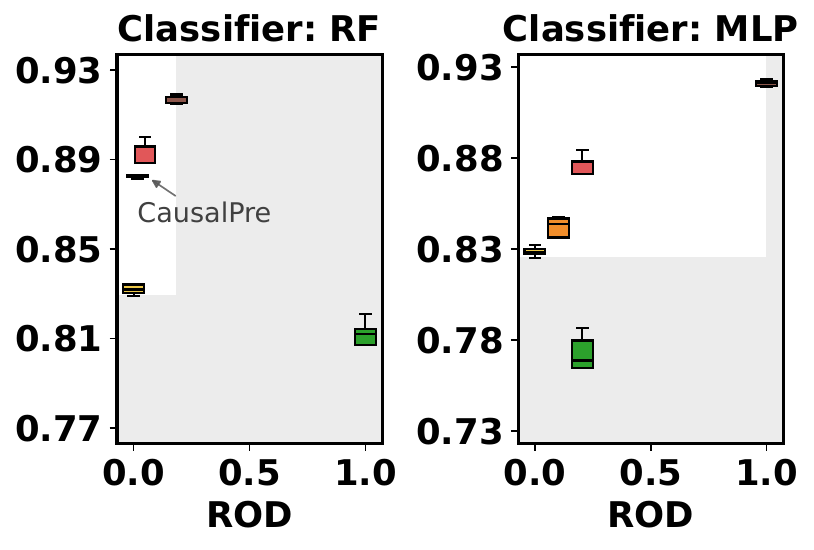}
            \vspace{-2mm} {\small (c) $\mathsf{Census\text{-}KDD}$}
        \end{minipage}
        \captionof{figure}{End-to-end performance under attribute ambiguity across three datasets, measured by AUC (utility) and ROD (fairness). Each box summarizes the 5-fold AUC values and the corresponding average ROD for one approach; higher AUC and lower ROD indicate better performance. Shaded regions indicate invalid results.}
        \label{fig: end2end_imperfect}
    \end{minipage}

    \vspace{6mm}

    \begin{minipage}[t]{\textwidth}
        \centering
        \begin{minipage}[t]{0.018\linewidth}
            \centering
            \hspace{0mm}\includegraphics[width=\linewidth, trim=0mm 0mm 0.4mm 0mm, clip]{figure/end-to-end_y-axis.pdf}
            \vspace{-2mm}
        \end{minipage}
        \begin{minipage}[t]{0.32\linewidth}
            \centering
            \hspace{-2mm}\includegraphics[width=\linewidth]{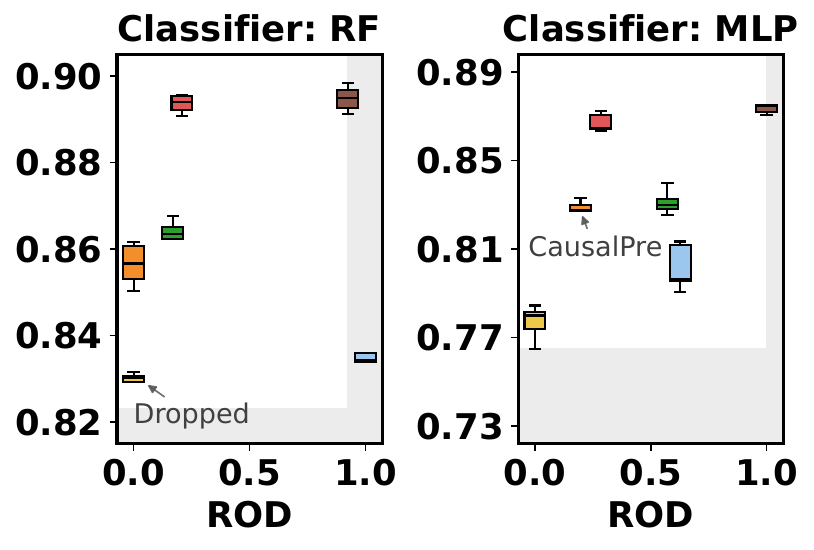}
            \vspace{-2mm} {\small (a) $\mathsf{Adult}$}
        \end{minipage}
        \begin{minipage}[t]{0.32\linewidth}
            \centering
            \hspace{-2mm}\includegraphics[width=\linewidth]{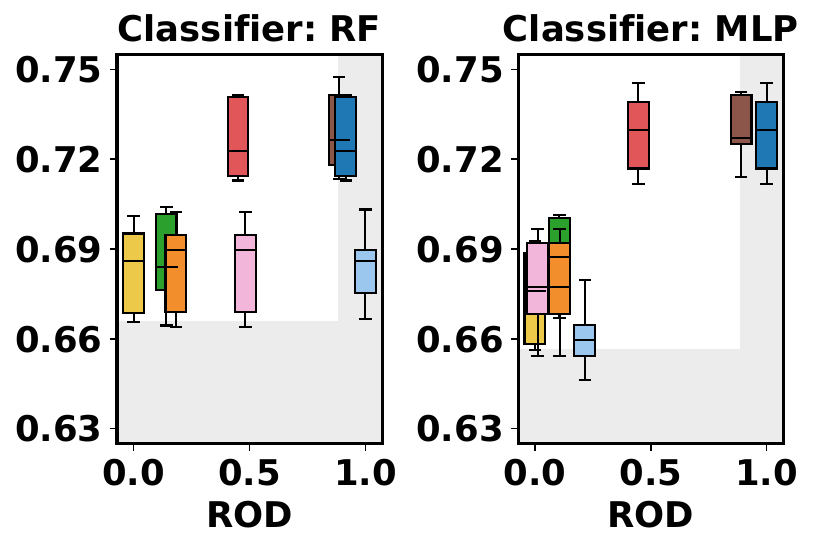}
            \vspace*{-2mm} {\small (b) $\mathsf{COMPAS}$}
        \end{minipage}
        \begin{minipage}[t]{0.32\linewidth}
            \centering
            \hspace{-2mm}\includegraphics[width=\linewidth]{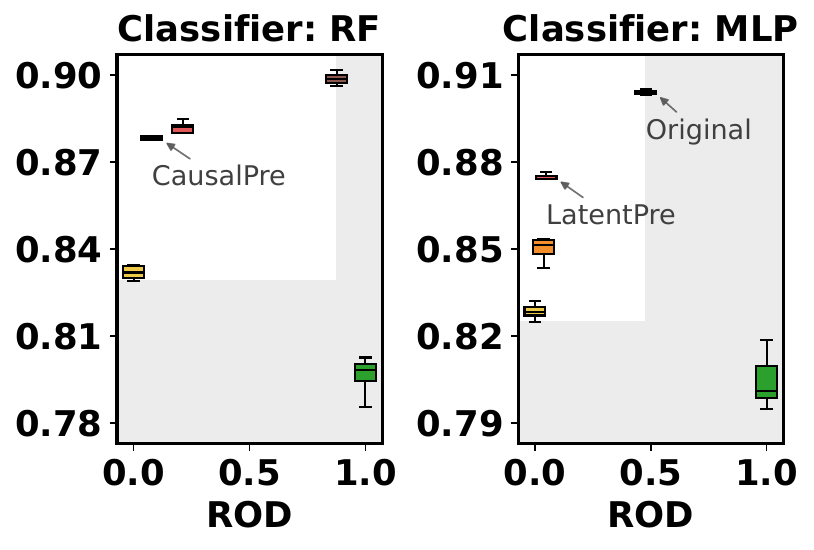}
            \vspace{-2mm} {\small (c) $\mathsf{Census\text{-}KDD}$}
        \end{minipage}
        \captionof{figure}{End-to-end performance under attribute absence across three datasets.}
        \label{fig: end2end_missing}
    \end{minipage}

    \vspace{6mm}

    \begin{minipage}[t]{\textwidth}
        \centering
        \begin{minipage}[t]{0.018\linewidth}
            \centering
            \hspace{0mm}\includegraphics[width=\linewidth, trim=0mm 0mm 0.4mm 0mm, clip]{figure/end-to-end_y-axis.pdf}
            \vspace{-2mm}
        \end{minipage}
        \begin{minipage}[t]{0.32\linewidth}
            \centering
            \hspace{-2mm}\includegraphics[width=\linewidth]{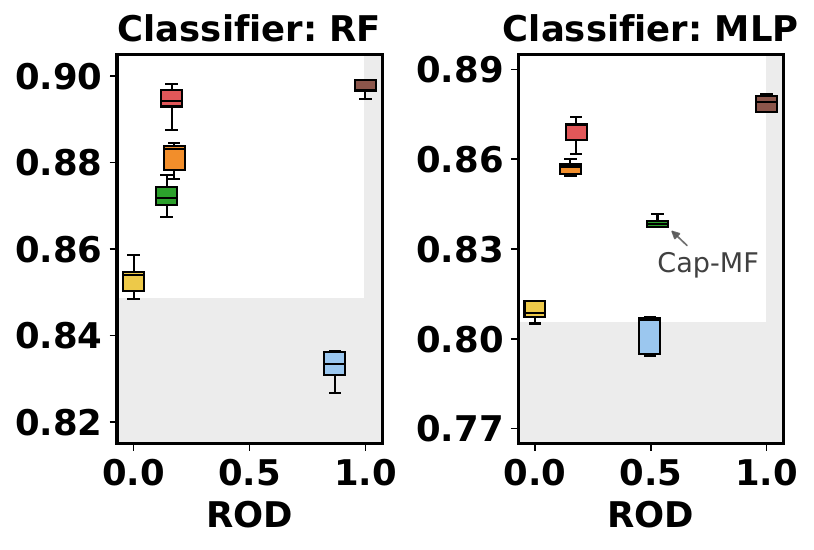}
            \vspace{-2mm} {\small (a) $\mathsf{Adult}$}
        \end{minipage}
        \begin{minipage}[t]{0.32\linewidth}
            \centering
            \hspace{-2mm}\includegraphics[width=\linewidth]{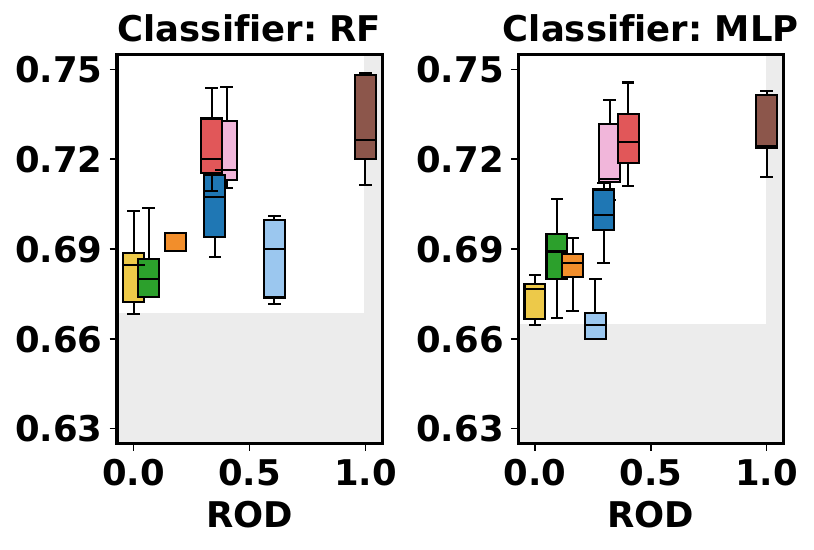}
            \vspace*{-2mm} {\small (b) $\mathsf{COMPAS}$}
        \end{minipage}
        \begin{minipage}[t]{0.32\linewidth}
            \centering
            \hspace{-2mm}\includegraphics[width=\linewidth]{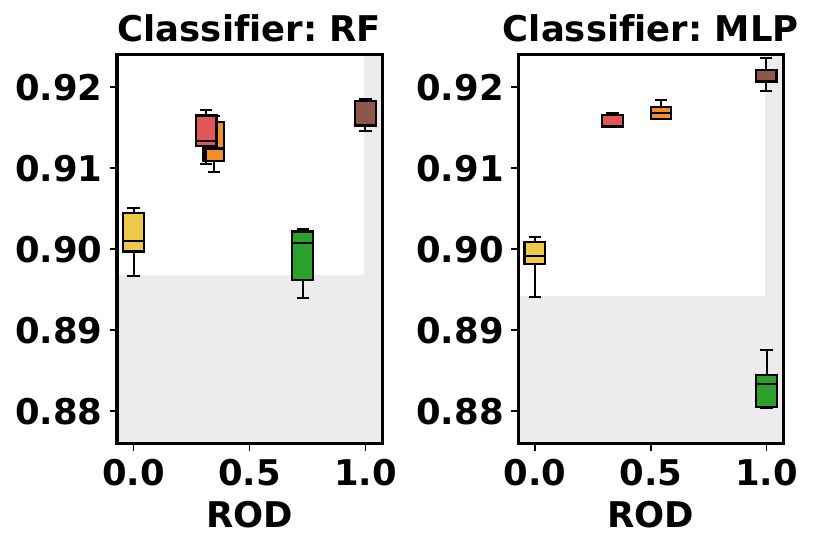}
            \vspace{-2mm} {\small (c) $\mathsf{Census\text{-}KDD}$}
        \end{minipage}
        \captionof{figure}{End-to-end performance under perfect attribute space across three datasets.}
        \label{fig: end2end_normal}
    \end{minipage}
\end{figure*}

\subsection{End-to-End Performance Evaluation} \label{subsec: exp-endtoend}

We evaluate the end-to-end performance of \latent\ under three scenarios: (i) attribute ambiguity, (ii) attribute absence, and (iii) perfect attribute space. Results are shown in Figures~\ref{fig: end2end_imperfect},~\ref{fig: end2end_missing}, and~\ref{fig: end2end_normal}, respectively. Performance is measured using box plots of 5-fold AUC and average ROD. Cap-MS is excluded from Census-KDD due to excessive runtime (over 24 hours). OTClean is omitted from Adult and Census-KDD because it runs out of memory on our machine. Default parameters for \latent\ are set as follows: the number of latent states $\tau=6$, $3$, and $6$ for Adult, COMPAS, and Census-KDD; the number of EM iterations $n=800$; and convergence threshold $\eta=0.001$. A method is deemed invalid if its fairness is worse than the unprocessed ``Original'' baseline or its utility lower than the ``Dropped'' case; these regions are shaded in gray. Key findings for each scenario are summarized below. 

\subsubsection{Attribute Ambiguity.}\label{subsec: exp-endtoend-ambiguity}
This setting evaluates robustness when certain utility-relevant but ambiguous attributes are conservatively treated as inadmissible. To simulate this scenario, we randomly reassign a subset of utility-relevant attributes from the admissible or additional set to the inadmissible set. Specifically, we reassign one attribute for the Adult and COMPAS datasets, and six for the larger Census-KDD dataset. We use a larger number for Census-KDD because it contains many structurally similar or semantically related attributes, which practitioners may be more inclined to conservatively label as inadmissible together.
Figure~\ref{fig: end2end_imperfect} reports the results. Across all datasets, \latent\ significantly outperforms the baselines in utility and consistently remains valid. This holds even on COMPAS, a low-dimensional dataset where a limited attribute space makes the causal structure more fragile and fairness-aware processing more challenging. Even in this setting, \latent\ remains robust, demonstrating the effectiveness of latent augmentation in recovering and reinforcing the underlying structure from sparse observations. In addition, \latent\ reduces average discrimination by 80\%, 46\%, and 77\%, with only 0.2\%, 0.5\%, and 4.2\% utility loss on datasets Adult, COMPAS, and Census-KDD, respectively. The fairness gain on COMPAS is relatively smaller, and we defer the explanation to Section~\ref{subsec: exp-endtoend-discussion}. Nevertheless, this result is still acceptable, as the utility loss is minimal and the fairness gain comes at negligible cost. On Census-KDD, performance declines slightly because we treat more attributes as imperfect to reflect a more challenging and realistic setting, yet \latent\ still achieves a 2.2\% average utility improvement over the best-performing baseline.

\subsubsection{Attribute Absence.}\label{subsec: exp-endtoend-absence}
This setting evaluates performance when some meaningful attributes are entirely missing from the dataset. To simulate this scenario, we randomly remove utility-relevant attributes from the admissible or additional set: one for Adult and COMPAS, and six for Census-KDD. The results in Figure~\ref{fig: end2end_missing} show that \latent\ substantially outperforms all baselines in utility while also reducing discrimination. Specifically, \latent\ lowers average discrimination by 75\%, 50\%, and 83\%, with only 0.4\%, 0.3\%, and 2.6\% utility loss on datasets Adult, COMPAS, and Census-KDD, respectively. The overall trend is consistent with the results observed under attribute ambiguity. 

\subsubsection{Perfect Attribute Space.}\label{subsec: exp-endtoend-perfect}
This setting evaluates performance when all attributes are correctly collected and specified, as in prior work~\cite{salimi2019interventional, zheng2025causalpre}. The results in Figure~\ref{fig: end2end_normal} show that \latent\ still achieves utility gains in most cases, even in this idealized setting. This indicates that attribute ambiguity or absence may naturally arise in practice, and \latent\ can effectively identify and correct such hidden imperfections through latent modeling. 
Note that although \latent\ sometimes attains the best ROD, introducing the latent attribute does not itself contribute to fairness. Fairness is ensured by explicitly removing all unfair causal pathways from sensitive or inadmissible attributes to the label. Meanwhile, introducing the latent attribute requires pruning certain dependencies among inadmissible attributes to ensure identifiability. This identifiability-driven pruning can also reduce spurious high-dimensional correlations that powerful classifiers might otherwise inadvertently exploit as unfair influence. This side effect is typically more pronounced in higher-dimensional datasets such as Census-KDD.

\begin{figure*}[htbp]
    \centering
    \begin{subfigure}{\linewidth}
      \includegraphics[width=\linewidth]{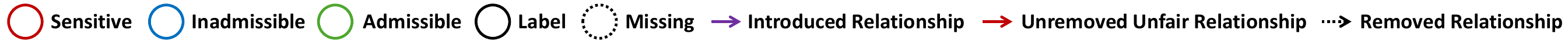}
    \end{subfigure} \\
    \vspace{2.5mm}
    \begin{subfigure}{0.148\linewidth}
      \includegraphics[width=\linewidth]{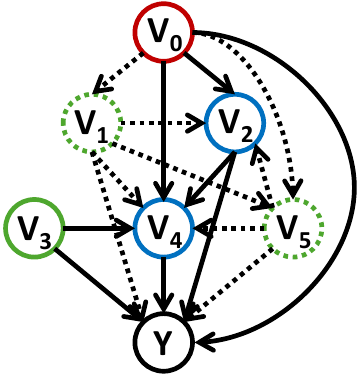}
      \caption{Missing.}
      \label{fig: dag_missing}
    \end{subfigure} \quad
    \begin{subfigure}{0.11\linewidth}
      \includegraphics[width=\linewidth]{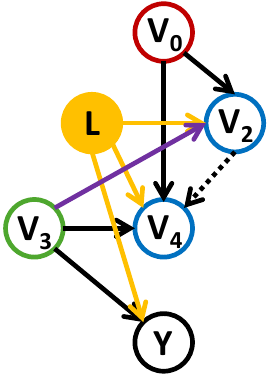}
      \caption{\latent.}
      \label{fig: dag_latent}
    \end{subfigure} \quad
    \begin{subfigure}{0.11\linewidth}
      \includegraphics[width=\linewidth]{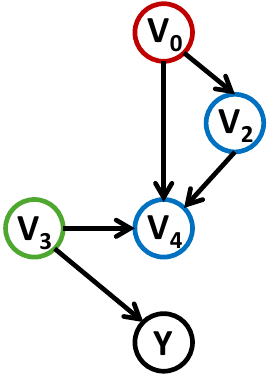}
      \caption{CausalPre.}
      \label{fig: dag_lazy}
    \end{subfigure} \quad
    \begin{subfigure}{0.11\linewidth}
      \includegraphics[width=\linewidth]{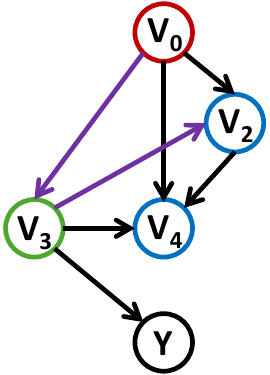}
      \caption{Cap-MF.}
      \label{fig: dag_mf}
    \end{subfigure} \quad
    \begin{subfigure}{0.126\linewidth}
      \includegraphics[width=\linewidth]{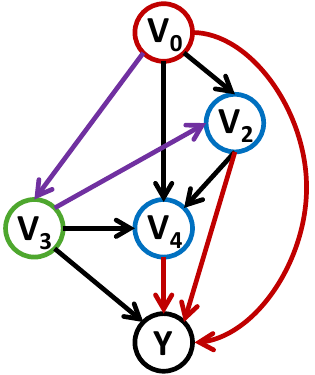}
      \caption{Cap-MS.}
      \label{fig: dag_ms}
    \end{subfigure} \quad
    \begin{subfigure}{0.11\linewidth}
      \includegraphics[width=\linewidth]{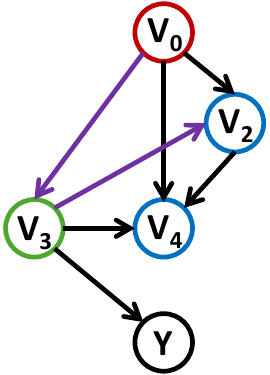}
      \caption{OTClean.}
      \label{fig: dag_ot}
    \end{subfigure} \quad
    \begin{subfigure}{0.126\linewidth}
      \includegraphics[width=\linewidth]{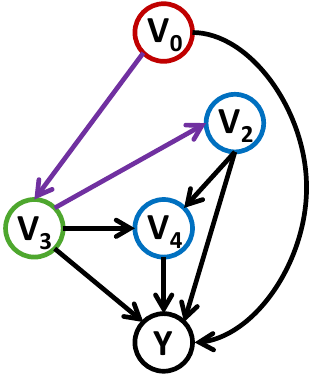}
      \caption{OTClean-RT.}
      \label{fig: dag_otrt}
    \end{subfigure} 
    \caption{Relationship recovery with $\mathcal{S}{=}\{V_0\}, \mathcal{I}{=}\{V_2, V_4\}, \mathcal{A}{=}\{V_1, V_3, V_5\}$, and missing attributes $\{V_1, V_5\}$.}
    \label{fig: dag}

\end{figure*}

\subsubsection{Additional Discussion} \label{subsec: exp-endtoend-discussion}
As noted, \latent\ yields smaller fairness improvements on COMPAS across all three scenarios. This outcome is expected given two inherent properties of the dataset. First, COMPAS exhibits strong correlations between utility-related and sensitive information, making the utility–fairness trade-off particularly pronounced. Although the introduced latent attribute is not causally influenced by the sensitive attribute (and we further evaluate their independence in Section~\ref{subsec: exp-indep}), powerful classifiers may still exploit spurious high-dimensional correlations in the data. This phenomenon is not unique to \latent: OTClean shows a similarly sharp utility–fairness trade-off on COMPAS; moreover, since OTClean is not designed for imperfect scenarios, its behavior in such settings is even less controllable. Second, COMPAS is relatively small, which makes the empirical distribution less stable and further increases the risk that powerful classifiers capture spurious correlations. To validate this explanation, we conduct an auxiliary experiment that enlarges COMPAS via repeated sampling (by increasing $k$ in Algorithm~\ref{algo: complete}, Line~\ref{line: sample k}). With the MLP classifier under this setting, \latent\ achieves a lower (thus better) ROD of 0.30, which is comparable to the gains observed on other datasets. In summary, the limited gains on COMPAS primarily reflect the dataset’s challenging characteristics rather than a weakness of \latent. Importantly, \latent\ still delivers meaningful fairness improvements with almost no sacrifice in utility, making the outcome practically satisfactory.

\begin{table}[t]
    \centering 
    \caption{Independence analysis of $L$ and $\mathcal{S}$.}
    \label{tab: indep_analysis_shorten}
    \begin{small}
    \begin{tabular}{c|ll|cc}
        \toprule
        Dataset & Scen. & Sens. Attr. & NMI & chi-square \\
        \midrule

        
        \multirow{9}{*}{\rotatebox[origin=c]{90}{Adult}}
        
        & \multirow{3}{*}{Amb.}
        & sex
        & $(3.1 \pm 0.8) \times 10^{-5}$
        & $0.71 \pm 0.12$ \\
        
        && race
        & $(1.7 \pm 0.2) \times 10^{-4}$
        & $0.74 \pm 0.14$ \\
        
        && sex + race
        & $(3.6 \pm 0.4) \times 10^{-4}$
        & $0.58 \pm 0.21$ \\
        \cmidrule(lr){2-5}
        
        & \multirow{3}{*}{Abs.}
        & sex
        & $(2.9 \pm 1.5) \times 10^{-5}$
        & $0.74 \pm 0.20$ \\
        
        && race
        & $(2.0 \pm 0.1) \times 10^{-4}$
        & $0.58 \pm 0.09$ \\
        
        && sex + race
        & $(3.8 \pm 0.3) \times 10^{-4}$
        & $0.51 \pm 0.16$ \\
        \cmidrule(lr){2-5}

        & \multirow{3}{*}{Perf.}
        & sex 
        & $(2.7 \pm 0.6) \times 10^{-5}$
        & $0.76 \pm 0.09$ \\
        
        && race
        & $(2.0 \pm 0.2) \times 10^{-4}$
        & $0.56 \pm 0.12$ \\
        
        && sex + race
        & $(3.8 \pm 0.4) \times 10^{-4}$
        & $0.53 \pm 0.22$ \\
        \midrule


        \multirow{9}{*}{\rotatebox[origin=c]{90}{COMPAS}}
        
        & \multirow{3}{*}{Amb.}
        & sex
        & $(5.9 \pm 3.0) \times 10^{-5}$
        & $0.74 \pm 0.13$ \\
        
        && race
        & $(2.3 \pm 0.5) \times 10^{-4}$
        & $0.75 \pm 0.10$ \\
        
        && sex + race
        & $(4.4 \pm 1.1) \times 10^{-4}$
        & $0.86 \pm 0.09$ \\
        \cmidrule(lr){2-5}
        
        & \multirow{3}{*}{Abs.}
        & sex
        & $(1.1 \pm 1.3) \times 10^{-4}$
        & $0.66 \pm 0.29$ \\
        
        && race
        & $(1.8 \pm 0.3) \times 10^{-4}$
        & $0.85 \pm 0.06$ \\
        
        && sex + race
        & $(4.8 \pm 1.1) \times 10^{-4}$
        & $0.82 \pm 0.11$ \\
        \cmidrule(lr){2-5}

        & \multirow{3}{*}{Perf.}
        & sex
        & $(3.5 \pm 2.4) \times 10^{-5}$
        & $0.83 \pm 0.10$ \\
        
        && race
        & $(2.2 \pm 0.5) \times 10^{-4}$
        & $0.76 \pm 0.09$ \\
        
        && sex + race
        & $(4.4 \pm 1.6) \times 10^{-4}$
        & $0.83 \pm 0.15$ \\
        \midrule
        

        \multirow{9}{*}{\rotatebox[origin=c]{90}{Census-KDD}}

        & \multirow{3}{*}{Amb.}
        & sex
        & $(2.0 \pm 1.0) \times 10^{-5}$
        & $0.28 \pm 0.30$ \\
        
        && race
        & $(4.1 \pm 1.3) \times 10^{-5}$
        & $0.73 \pm 0.22$ \\
        
        && sex + race
        & $(8.6 \pm 2.3) \times 10^{-5}$
        & $0.57 \pm 0.30$ \\
        \cmidrule(lr){2-5}
        
        & \multirow{3}{*}{Abs.}
        & sex
        & $(2.2 \pm 0.5) \times 10^{-5}$
        & $0.12 \pm 0.08$ \\
        
        && race
        & $(4.3 \pm 0.5) \times 10^{-5}$
        & $0.61 \pm 0.13$ \\
        
        && sex + race
        & $(8.5 \pm 0.9) \times 10^{-5}$
        & $0.43 \pm 0.17$ \\
        \cmidrule(lr){2-5}

        & \multirow{3}{*}{Perf.}
        & sex
        & $(1.2 \pm 0.1) \times 10^{-5}$
        & $0.31 \pm 0.06$ \\
        
        && race
        & $(2.7 \pm 0.3) \times 10^{-5}$
        & $0.86 \pm 0.07$ \\
        
        && sex + race
        & $(6.2 \pm 0.3) \times 10^{-5}$
        & $0.70 \pm 0.08$ \\
        \bottomrule
    \end{tabular}
    \end{small}
\end{table}

\subsection{Missing Relationship Recovery} \label{subsec: exp-recovery}

To evaluate how well \latent\ and the baselines preserve potential causal relationships when some attributes are missing, we generate a synthetic dataset from the DAG in Figure~\ref{fig: dag_missing}. The dataset contains 50,000 records and 7 attributes. Note that the attribute space is deliberately kept small to make DAG visualization feasible. 
\begin{figure}
    \centering
    \hspace{1.7mm} \includegraphics[width=0.345\linewidth]{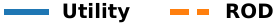}
    \vspace{-2mm}
\end{figure}

\begin{figure}[t]
    \centering
    \begin{subfigure}[t]{0.49\linewidth}
        \centering
        \includegraphics[width=\linewidth, trim=2.1mm 2.5mm 2.1mm 0mm, clip]{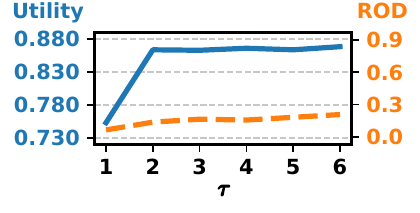}
        \par\vspace{-0.5mm}
        \caption{Attribute ambiguity.}
        \label{fig: param_l_imperfect}
    \end{subfigure}
    \begin{subfigure}[t]{0.49\linewidth}
        \centering
        \includegraphics[width=\linewidth, trim=2.1mm 2.5mm 2.1mm 0mm, clip]{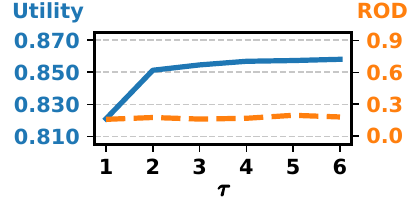}
        \par\vspace{-0.5mm}
        \caption{Attribute absence.}
        \label{fig: param_l_missing}
    \end{subfigure}
    \vspace{-1mm}
    \caption{Varying the number of latent states $\tau$.}
    \label{fig: param_l}
\end{figure}

\vspace{1mm}

\begin{figure}[t]
    \centering
    \begin{subfigure}[t]{0.49\linewidth}
        \centering
        \includegraphics[width=\linewidth, trim=2.1mm 2.5mm 2.1mm 0mm, clip]{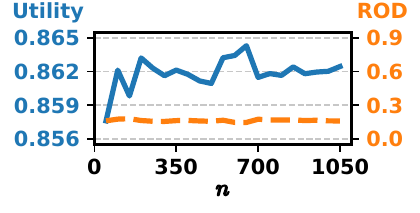}
        \par\vspace{-0.5mm}
        \caption{Attribute ambiguity.}
        \label{fig: param_n_imperfect}
    \end{subfigure}
    \begin{subfigure}[t]{0.49\linewidth}
        \centering
        \includegraphics[width=\linewidth, trim=2.1mm 2.5mm 2.1mm 0mm, clip]{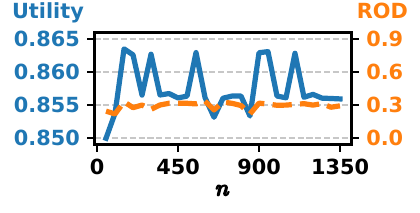}
        \par\vspace{-0.5mm}
        \caption{Attribute absence.}
        \label{fig: param_n_missing}
    \end{subfigure} 
    \vspace{-1mm}
    \caption{Varying the number of iterations $n$.}
    \label{fig: param_n}
\end{figure}

\vspace{1mm}

\begin{figure}[t]
    \centering
    \begin{subfigure}[t]{0.49\linewidth}
        \centering
        \includegraphics[width=\linewidth, trim=2.1mm 2.3mm 2.1mm 0mm, clip]{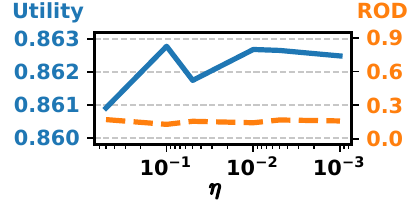}
        \par\vspace{-0.5mm}
        \caption{Attribute ambiguity.}
        \label{fig: param_thr_imperfect}
    \end{subfigure} 
    \begin{subfigure}[t]{0.49\linewidth}
        \centering
        \includegraphics[width=\linewidth, trim=2.1mm 2.3mm 2.1mm 0mm, clip]{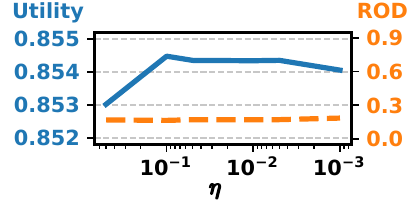}
        \par\vspace{-0.5mm}
        \caption{Attribute absence.}
        \label{fig: param_thr_missing}
    \end{subfigure} 
    \vspace{-1mm}
    \caption{Varying the convergence threshold $\eta$.}
    \label{fig: param_thr}
\end{figure}
On this dataset, we first randomly assign attributes to reasonable roles, such as inadmissible or admissible, and then randomly remove several attributes from the admissible attribute set to simulate scenarios where some meaningful and usable information is lost. We next apply \latent\ and the baselines to the resulting dataset and reconstruct a DAG from each processed output using the Python libraries \textit{pgmpy.estimators} and \textit{networkx}. These recovered DAGs provide a direct view of which causal relationships are preserved after processing.

Figure~\ref{fig: dag} shows one representative setting, where $\{V_0\}$ is sensitive, $\{V_2,V_4\}$ are inadmissible, $\{V_1,V_3,V_5\}$ are admissible, and $Y$ is the label, with $V_1$ and $V_5$ removed to simulate missing attributes. As observed, CausalPre in Figure~\ref{fig: dag_lazy} removes all unfair pathways and preserves observed $V_3$ as a decision factor. OTClean in Figure~\ref{fig: dag_ot} and Cap-MF in Figure~\ref{fig: dag_mf} yield a similar structure but introduce two additional edges, which are highlighted in purple. In contrast, Cap-MS fails to block all unfair causal pathways, as shown in Figure~\ref{fig: dag_ms}, because it substantially distorts the data distribution during processing. This also explains why its results are always invalid in Section~\ref{subsec: exp-endtoend}. Figure~\ref{fig: dag_otrt} reports OTClean-RT. Since it targets a different processing constraint, we evaluate it under a corresponding criterion that focuses on severing dependence between sensitive and inadmissible attributes. Under this criterion, OTClean-RT successfully breaks this dependence, but it still introduces the same two extra edges as OTClean and Cap-MF. Overall, none of the baselines recover the potential utility-relevant information associated with the missing attributes.

\latent\ exhibits a different behavior by explicitly recovering the lost information. As shown in Figure~\ref{fig: dag_latent}, it introduces a latent attribute $L$ that contributes to prediction alongside $V_3$. This latent attribute captures the influence of the missing attributes $V_1$ and $V_5$ by reinstating their direct relationships with other attributes such as $V_2$ and $V_4$, and by passing the recovered information to the label. Consequently, \latent\ restores the missing causal pathways, as if $V_1$ and $V_5$ were still present. Achieving this recovery requires a deliberate trade-off. To preserve more utility-relevant information near the label, \latent\ modifies certain relationships around the inadmissible attributes: it removes the edge between $V_2$ and $V_4$ for identifiability and introduces an additional edge $V_3 \rightarrow V_2$ to support a coarser but tractable policy. Strictly speaking, $V_3$ should be treated as an additional attribute rather than an admissible one, in which case the edge $V_3 \rightarrow V_2$ would not be introduced. We treat $V_3$ as admissible here only to allow the baselines to operate more effectively on this simple dataset. This trade-off proves beneficial: as evidenced by our further end-to-end evaluations, \latent\ consistently outperforms the baselines, achieving an average utility improvement of 22\% across both classifiers.

\begin{figure}[t]
    \centering

    \begin{minipage}[t]{0.04\linewidth}
        \centering
        \begin{subfigure}[t]{0.9\linewidth}
            \hspace*{-2.5mm}\includegraphics[width=\linewidth, trim=0mm -2mm 0mm 0mm, clip]{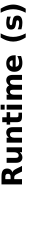}
        \end{subfigure}
        
        \vspace{7.5mm}
        
        \begin{subfigure}[t]{\linewidth}
            \hspace*{-2.5mm}\includegraphics[width=0.9\linewidth, trim=0mm -2mm 0mm 0mm, clip]{figure/runtime_legend.pdf}
        \end{subfigure}
        
        \vspace{7.5mm}
        
        \begin{subfigure}[t]{\linewidth}
            \hspace*{-2.5mm}\includegraphics[width=0.9\linewidth, trim=0mm -2mm 0mm 0mm, clip]{figure/runtime_legend.pdf}
        \end{subfigure}
    \end{minipage}
    %
    %
    \begin{minipage}[t]{0.37\linewidth}
        \centering
        \begin{subfigure}[t]{\linewidth}
            \centering
            \hspace*{-5.5mm}\includegraphics[width=\linewidth]{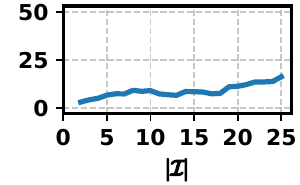}
            \vspace{-1mm}
            \caption{}
            \label{fig: staged-ident-i}
        \end{subfigure}
        
        \vspace{3mm}
        
        \begin{subfigure}[t]{\linewidth}
            \centering
            \hspace*{-5.5mm}\includegraphics[width=\linewidth]{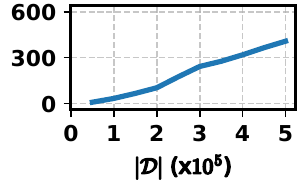}
            \vspace{-1mm}
            \caption{}
            \label{fig: staged-ident-n}
        \end{subfigure}
        
        \vspace{3mm}
        
        \begin{subfigure}[t]{\linewidth}
            \centering
            \hspace*{-5.5mm}\includegraphics[width=\linewidth]{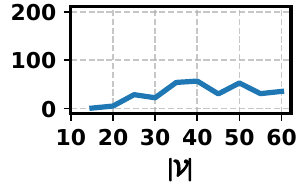}
            \vspace{-1mm}
            \caption{}
            \label{fig: staged-ident-v}
        \end{subfigure}
    \end{minipage}
    \hspace{2mm}
    %
    %
    \begin{minipage}[t]{0.37\linewidth}
        \centering
        \begin{subfigure}[t]{\linewidth}
            \centering
            \hspace*{-5.5mm}\includegraphics[width=\linewidth]{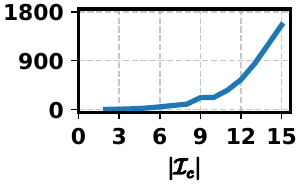}
            \vspace{-1mm}
            \caption{}
            \label{fig: staged-part-i}
        \end{subfigure}
        
        \vspace{3mm}
        
        \begin{subfigure}[t]{\linewidth}
            \centering
            \hspace*{-5.5mm}\includegraphics[width=\linewidth]{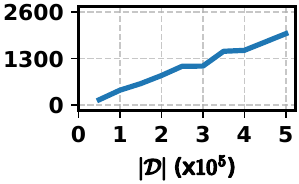}
            \vspace{-1mm}
            \caption{}
            \label{fig: staged-part-n}
        \end{subfigure}
        
        \vspace{3mm}
        
        \begin{subfigure}[t]{\linewidth}
            \centering
            \hspace*{-5.5mm}\includegraphics[width=\linewidth]{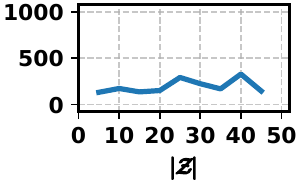}
            \vspace{-1mm}
            \caption{}
            \label{fig: staged-part-v}
        \end{subfigure}
    \end{minipage}

    \caption{Staged runtime analysis for $\mathcal{I}_c$ identification algorithm ((a)--(c)) and $\mathcal{I}_c$ partition algorithm ((d)--(f)).}
    \label{fig: staged}
\end{figure}

\subsection{Verifying Latent--Sensitive Independence} \label{subsec: exp-indep}

As mentioned in Section~\ref{subsec: method-framework}, the latent attribute $L$ is designed to be independent of the sensitive attributes $\mathcal{S}$. To empirically validate this property, we test whether the learned $L$ exhibits any residual dependence on $\mathcal{S}$ from both information-theoretic and statistical perspectives, using Normalized Mutual Information (NMI) and the chi-square test. NMI measures the strength of dependence, where a value of $0$ indicates independence and values below $10^{-3}$ are commonly treated as negligible in practice. The chi-square test assesses statistical evidence of dependence: the resulting $p$-values above $0.01$ indicate no statistically significant dependence. Table~\ref{tab: indep_analysis_shorten} presents the 5-fold results on three real-world datasets under two imperfect scenarios as well as the perfect scenario. Since each dataset includes two sensitive attributes, \texttt{sex} and \texttt{race}, we test dependence with respect to each attribute separately as well as their intersectional groups. Across all settings, NMI values are near zero and chi-square $p$-values consistently exceed the threshold. Both results indicate no detectable dependence between $L$ and $\mathcal{S}$, which is consistent with our design principle.

\subsection{Parameter Analysis} \label{subsec: exp-param}

We study the sensitivity of \latent\ to its key hyperparameters. For brevity, we report results on the representative Adult dataset with the MLP classifier; similar trends hold across other datasets and classifiers. As described in Algorithm~\ref{algo: complete}, \latent\ has three parameters: the number of latent states $\tau$, the maximum number of iterations $n$, and the convergence threshold $\eta$. Unless otherwise stated, when varying one hyperparameter, we fix the other two to $\tau=6$, $n=800$, and $\eta=0.001$.

\vspace{1mm}
\noindent\textbf{Varying $\tau$.}
Figure~\ref{fig: param_l} shows the impact of varying the number of latent states $\tau$. When $\tau=1$, all records share a single latent value, which is equivalent to having no latent attribute. As $\tau$ increases and captures more variation, utility improves consistently, with gains of 14.66\% and 3.68\% under the ambiguous and absent attribute settings, respectively. The gains are more pronounced under attribute ambiguity, since the ambiguous attributes remain in the dataset, making their influence easier to recover in the early stages. As $\tau$ continues to grow, utility improves marginally while ROD remains stable. Due to identifiability constraints, $\tau$ must lie within a bounded range; for Adult, we set the maximum to 6.

\vspace{1mm}
\noindent\textbf{Varying $n$ and $\eta$.}
Figure~\ref{fig: param_n} illustrates performance changes with the number of iterations $n$. 
Across both scenarios, utility increases in the early iterations and then stabilizes, while fairness remains steady throughout. A similar pattern is observed on the convergence threshold $\eta$, as shown in Figure~\ref{fig: param_thr}, within the expected noise levels.

\begin{figure*}[t]
    \centering
    \includegraphics[width=0.65\linewidth, keepaspectratio]{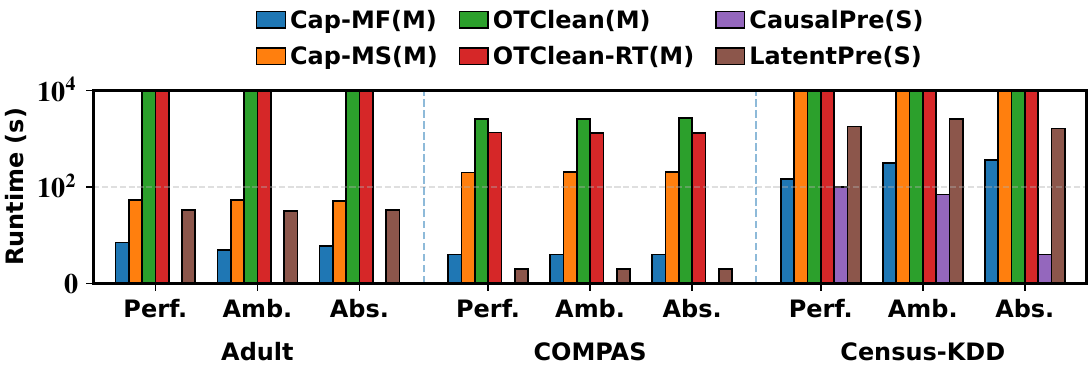}
    \caption{Runtime comparison. (M: execution with 64 threads; S: execution with a single thread.)}
    \vspace{2mm}
    \label{fig: runtime}
\end{figure*}

\subsection{Runtime Evaluation} \label{subsec: exp-time}

\subsubsection{Staged}\label{subsec: exp-time-staged}%
To evaluate the efficiency of the identification and partition algorithms proposed in Sections~\ref{subsec: method-identification} and~\ref{subsec: method-partition}, we study how runtime scales with key dataset properties using synthetic data. For identification, we vary one factor at a time among the inadmissible-set size $|\mathcal{I}|$, the dataset size $|\mathcal{D}|$, and the number of attributes $|\mathcal{V}|$, as shown in Figures~\ref{fig: staged-ident-i}--\ref{fig: staged-ident-v}, while fixing the other two at $|\mathcal{I}|{=}10$, $|\mathcal{D}|{=}50{,}000$, and $|\mathcal{V}|{=}30$. As shown, increasing $|\mathcal{I}|$ or $|\mathcal{V}|$ leads only to a mild runtime increase with small fluctuations. This indicates that CI tests executed in Lines~\ref{line: g start}--\ref{line: g end} of Algorithm~\ref{algo: identification} are far fewer than the worst-case count. In practice, the procedure typically terminates early rather than enumerating all conditioning sets. When $|\mathcal{D}|$ increases, the runtime grows linearly, as each CI test becomes more expensive with more records. For partition, we similarly vary one factor at a time among $|\mathcal{I}_c|$, $|\mathcal{D}|$, and the conditioning-set size $|\mathcal{Z}|{=}| \mathcal{S}\cup\mathcal{I}_o\cup\mathcal{A}|$, as shown in Figures~\ref{fig: staged-part-i}--\ref{fig: staged-part-v}, while fixing the remaining two at $|\mathcal{I}_c|{=}10$, $|\mathcal{D}|{=}50{,}000$, and $|\mathcal{Z}|{=}20$. The partition runtime scales quadratically with $|\mathcal{I}_c|$, grows linearly with $|\mathcal{D}|$, and shows only mild sensitivity to $|\mathcal{Z}|$. Overall, the observed trends are well below the theoretical worst-case bound discussed in Section~\ref{subsec: method-partition} and support practical scalability.

\subsubsection{Overall}\label{subsec: exp-time-overall}
Figure~\ref{fig: runtime} reports the runtime of \latent\ and the baselines. \latent\ and CausalPre use a single thread, while the others use 64 threads. Bars reaching the top indicate failure to complete. CausalPre is sometimes not visible because its runtime falls below the plotted scale, under 2 seconds. Overall, among the six approaches, \latent\ typically ranks third, after CausalPre and Cap-MF. In the most challenging setting, Census-KDD under the ambiguous scenario, \latent\ completes in 2{,}573 seconds, which remains acceptable for an offline pre-processing task.

\balance
\section{Related Work} \label{sec: related}

Our work focuses on justifiable fairness and is, to our knowledge, the first to address fair data pre-processing under an imperfect attribute space. All existing works assume that the available attributes are complete and correctly specified (as admissible, inadmissible, etc.), and are therefore not designed to tackle this new challenge. 

Salimi et al. \cite{salimi2019interventional} are the first to formalize the fairness-aware database repair problem and define the notion of justifiable fairness. Their framework, Capuchin, aims to achieve fairness while minimizing the number of record modifications. Building upon Capuchin, OTClean~\cite{pirhadi2024otclean} further controls distributional distortion via optimal transport, seeking to minimize distribution shift in addition to enforcing fairness. More recently, CausalPre~\cite{zheng2025causalpre} adopts a two-step strategy: it efficiently learns causally fair relationships from scratch and applies them to refine the dataset, ensuring fairness while preserving the overall statistical character. 

Several other pre-processing techniques~\cite{pirhadi2024otclean, pujol2023prefair, salazar2021automated, galhotra2022causal, robertson2025fairpfn} also adopt causal fairness, but target different problems. PreFair~\cite{pujol2023prefair} also builds on justifiable fairness, but does not consider causality; instead, it focuses on incorporating fairness constraints into private data synthesis, which differs from our repair-oriented framework. FairPFN~\cite{robertson2025fairpfn} studies fairness-aware pre-training with a fair data generation module. It aims to eliminate the total causal effect of a single binary sensitive attribute, and is therefore not directly comparable to our setting. FairExp~\cite{salazar2021automated} and SeqSel~\cite{galhotra2022causal}, on the other hand, focus on fair feature selection rather than data adjustment. 
Beyond causal fairness, several pre-processing methods~\cite{feldman2015certifying, nabi2018fair, zemel2013learning, xu2018fairgan, kamiran2012data, xiong2024fairwasp, calmon2017optimized, gordaliza2019obtaining, van2021decaf} instead target associational fairness, especially the notion of demographic parity. Since they enforce only associational constraints, they generally cannot guarantee causal fairness~\cite{salimi2019interventional}.

Finally, some works~\cite{tian2025towards, ma2023learning, russell2017worlds, kusner2017counterfactual} also attempt to explore unobserved signals from the data. EXCO~\cite{tian2025towards}, for example, introduces latent attributes to represent exogenous background factors that influence the sensitive attribute and uses them to develop a fair predictor. CLAIRE~\cite{ma2023learning} does not explicitly model latent attributes but instead learns fair representations with a variational autoencoder. These methods operate during model training, whereas our work addresses fairness through the pre-processing stage.

\section{Conclusion} \label{sec: conclusion}

In this paper, we re-examined the problem of fair data pre-processing under an imperfect attribute space. Our framework, \latent{}, augments the fairness policy with latent attributes that guarantee identifiability, recover valid missing signals, and block inadmissible influence. Guided by this policy, the raw data is adjusted to satisfy justifiable fairness while preserving utility. Extensive experiments show that \latent{} is robust across a variety of real-world imperfect settings and achieves strong fairness-utility trade-offs.

\clearpage

\bibliographystyle{ACM-Reference-Format}
\bibliography{main}

\end{document}